\documentclass[11pt]{article}

\usepackage{graphicx}
\usepackage{amsmath}
\usepackage{amsfonts}              
\usepackage{amsthm}                
\usepackage[margin=1in]{geometry}
\usepackage{hyperref}

\newtheorem{theorem}{Theorem}
\newtheorem{lemma}{Lemma}
\newtheorem{claim}{Claim}

\newtheorem{definition}{Definition}
\newtheorem{rem}{Remark}

\DeclareMathOperator*{\argmax}{argmax}
\DeclareMathOperator{\Line}{\mathsf{line}}

\DeclareMathOperator{\Surface}{\mathsf{surface}}

\DeclareMathOperator{\Star}{\mathsf{star}}
\DeclareMathOperator{\Deg}{\mathsf{deg}}
\DeclareMathOperator*{\E}{\mathbb{E}}

\DeclareMathOperator{\SA}{\mathsf{SA}}
\DeclareMathOperator{\calc}{\mathcal{C}}
\DeclareMathOperator{\Nbr}{\mathsf{Nbr}}

\DeclareMathOperator{\Frac}{\mathsf{Frac}}

\newcommand{\C}{\mathcal{C}}

\begin{document}

\title{Inapproximability of $H$-Transversal/Packing\footnote{It is an expanded, generalized, and refocused version of our earlier unpublished manuscript~\cite{GL14b} and our conference version that appears in the proceedings of APPROX 15.}}

\author{
Venkatesan Guruswami\thanks{Supported in part by NSF grant CCF-1115525. {\tt guruswami@cmu.edu} } \and
Euiwoong Lee\thanks{Supported by a Samsung Fellowship and NSF CCF-1115525. {\tt euiwoonl@cs.cmu.edu} }}

\date{Computer Science Department \\ Carnegie Mellon University \\ Pittsburgh, PA 15213.}

\maketitle

\begin{abstract}

Given an undirected graph $G = (V_G, E_G)$ and a fixed ``pattern" graph $H = (V_H, E_H)$ with $k$ vertices, we consider the $H$-Transversal and $H$-Packing problems. The former asks to find the smallest $S \subseteq V_G$ such that the  subgraph induced by $V_G \setminus S$ does not have $H$ as a subgraph, and the latter asks to find the maximum number of pairwise disjoint $k$-subsets $S_1, ..., S_m \subseteq V_G$ such that the subgraph induced by each $S_i$ has $H$ as a subgraph. 

We prove that if $H$ is 2-connected, $H$-Transversal and $H$-Packing are almost as hard to approximate as general $k$-Hypergraph Vertex Cover and $k$-Set Packing, so it is NP-hard to approximate them within a factor of $\Omega (k)$ and $\widetilde \Omega (k)$ respectively. 
We also show that there is a 1-connected $H$ where $H$-Transversal admits an $O(\log k)$-approximation algorithm, so that the connectivity requirement cannot be relaxed from 2 to 1. 
For a special case of $H$-Transversal where $H$ is a (family of) cycles, 
we mention the implication of our result to the related Feedback Vertex Set problem, and give a different hardness proof for directed graphs.

\end{abstract}



\thispagestyle{empty}
\newpage

\section {Introduction}
Given a collection of subsets $S_1, ..., S_m$ of the underlying set $U$, 
the {\em Set Transversal} problem asks to find the smallest subset of $U$ that intersects every $S_i$,
and the {\em Set Packing} problem asks to find the largest subcollection $S_{i_1}, ..., S_{i_{m'}}$ which are pairwise disjoint.\footnote{These problems are called many different names in the literature. Set Transversal is also called Hypergraph Vertex Cover, Set Cover (of the dual set system), and Hitting Set. Set Packing is also called Hypergraph Matching. We try to use Transversal / Packing unless another name is established in the literature (e.g. $k$-Hypergraph Vertex Cover).}
It is clear that optimum of the former is always at least that of the latter (i.e. weak duality holds).
Studying the (approximate) reverse direction of the inequality (i.e. strong duality) as well as the complexity of both problems 
for many interesting classes of set systems is arguably the most studied paradigm in combinatorial optimization.

This work focuses on set systems where the size of each set is bounded by a constant $k$. 
With this restriction, Set Transversal and Set Packing are known as $k$-Hypergraph Vertex cover ($k$-HVC) and $k$-Set Packing ($k$-SP), respectively.
This assumption significantly simplifies the problem since there are at most $n^{k}$ sets. 
While there is a simple factor $k$-approximation algorithm for both problems, it is NP-hard to approximate $k$-HVC and $k$-SP within a factor less than $k - 1$~\cite{DGKR05} and $O(\frac{k}{\log k})$~\cite{HSS06} respectively. 

Given a large graph $G = (V_G, E_G)$ and a fixed graph $H = (V_H, E_H)$ with $k$ vertices, 
one of the natural attempts to further restrict set systems is to set $U = V_G$, and take the collection of subsets to be all {\em copies} of $H$ in $G$ (formally defined in the next subsection). 
This natural representation in graphs often results in a deeper understanding of the underlying structure and better algorithms, 
with Maximum Matching ($H = K_2$) being the most well-known example. 
Kirkpatrick and Hell~\cite{KH83} proved that Maximum Matching is essentially the only case where $H$-Packing can be solved exactly in polynomial time --- unless $H$ is the union of isolated vertices and edges, it is NP-hard to decide whether $V_G$ can be partitioned into $k$-subsets each inducing a subgraph containing $H$. 
A similar characterization for the edge version (i.e. $U = E_G$) was obtained much later by Dor and Tarsi~\cite{DT97}.

We extend these results by studying the approximability of $H$-Transversal and $H$-Packing. We use the term {\em strong inapproximability} to denote NP-hardness of approximation within a factor $\Omega(k / polylog(k))$.
We give a simple sufficient condition that implies strong inapproximability --- if $H$ is 2-vertex connected, $H$-Transversal and $H$-Packing are almost as hard to approximate as $k$-HVC and $k$-SP. 
We also show that there is a 1-connected $H$ where $H$-Transversal admits an $O(\log k)$-approximation algorithm, so 1-connectivity is not sufficient for strong inapproximability for $H$-Transversal.
It is an interesting open problem whether 1-connectivity is enough to imply strong inapproximability of $H$-Packing, or there is a class of connected graphs where $H$-Packing admits a significantly nontrivial approximation algorithm (e.g. factor $k^{\epsilon}$ for some $\epsilon < 1$). 

Our results give an unified answer to questions left open in many independent works studying a special case where $H$ is a cycle or clique, and raises some new open questions. 
In the subsequent subsections, we state our main results, review related work, and state potential future directions.

\subsection{Problems and Our Results}
\label{subsec:results}

Given an undirected graphs $G = (V_G, E_G)$ and $H = (V_H, E_H)$ with $|V_H| = k$, we define the following problems.
\begin{itemize}
\item $H$-Transversal asks to find the smallest $F \subseteq V_G$ such that the subgraph of $G$ induced by $V_G \setminus F$ does not have $H$ as a subgraph.
\item $H$-Packing asks to find the maximum number of pairwise disjoint $k$-subsets of $S_1, ..., S_m$ of $V_G$ such that the subgraph induced by each $S_i$ has $H$ as a subgraph. 
\end{itemize}

Our main result states that 2-connectivity of $H$ is sufficient to make $H$-Transversal and $H$-Packing hard to approximate.

\begin{theorem}
\label{thm:main_v}
If $H$ is a 2-vertex connected with $k$ vertices, unless $\mathsf{NP} \subseteq \mathsf{BPP}$, no polynomial time algorithm approximates $H$-Transversal within a factor better than $k - 1$,
and $H$-Packing within a factor better than $\Omega(\frac{k}{\log^7 k})$. 
\end{theorem}

Let $k$-Star denote $K_{1, k - 1}$, the complete bipartite graph with 1 and $k-1$ vertices on each side. 
The following theorem shows that $k$-Star Transversal admits a good approximation algorithm, 
so the assumption of 2-connectedness in Theorem~\ref{thm:main_v} is required for strong inapproximability of $H$-Transversal.

\begin{theorem}
\label{thm:star}
$k$-Star Transversal can be approximated within a factor of $O(\log k)$ in polynomial time.
\end{theorem}

This algorithmic result matches $\Omega(\log k)$-hardness of $k$-Star Transversal via a simple reduction from Minimum Dominating Set on degree-$k$ graphs~\cite{CC08}. 
This problem has the following equivalent but more natural interpretation: given a graph $G = (V_G, E_G)$, find the smallest $F \subseteq V_G$ such that the subgraph induced by $V_G \setminus F$ has maximum degree at most $k - 2$. Our algorithm, which uses iterative roundings of 2-rounds of Sherali-Adams hierarchy of linear programming (LP) followed by a simple greedy algorithm for {\em Constrained Set Cover}, is also interesting in its own right, but we defer the details to Appendix~\ref{sec:star}.

Our hardness results for transversal problems rely on hardness of $k$-HVC which is NP-hard to approximate within a factor better than $k - 1$~\cite{DGKR05}. 
Our hardness results for packing problems rely on hardness of Maximum Independent Set on graphs with maximum degree $k$ and girth strictly {\em greater than} $g$ (MIS-$k$-$g$). Almost tight inapproximability of MIS on graphs with maximum degree $k$ (MIS-$k$) is recently proved in Chan~\cite{Chan13}, which rules out an approximation algorithm with ratio better than $\Omega(\frac{k}{\log^4 k})$. 
We are able to extend his result to MIS-$k$-$g$ with losing only a polylogarithmic factor. All applications in this work require $g = \Theta(k)$. 
\begin{theorem}
\label{thm:mis}
For any constants $k$ and $g$, unless $\mathsf{NP} \subseteq \mathsf{BPP}$, no polynomial time algorithm approximates MIS-$k$-$g$ within a factor of $\Omega(\frac{k}{\log^7 k})$.
\end{theorem}

We remark that assuming the Unique Games Conjecture (UGC) slightly improves our hardness ratios through better hardness of $k$-HVC~\cite{KR08} and MIS-$k$~\cite{AKS09}, and even simplifies the proof for some problems (e.g. $k$-Clique Transversal) through structured hardness of $k$-HVC~\cite{BK10}. Indeed, an earlier (unpublished) version of this work~\cite{GL14b} relied on the UGC to prove that MIS-$k$-$k$ is hard to approximate within a factor of $\Omega(\frac{k}{\log^4 k})$, while only giving $\widetilde\Omega(\sqrt{k})$-factor hardness without it.  
Now that we obtain almost matching hardness, we focus on proving hardness results without the UGC.

\subsection{Related Work and Special Cases}

After the aforementioned work characterizing those pattern graphs $H$ admitting the existence of a polynomial-time exact algorithm for $H$-Packing~\cite{KH83,DT97}, Lund and Yannakakis~\cite{LY93} studied the maximization version of $H$-Transversal (i.e. find the largest $V' \subseteq V_G$ such that the subgraph induced by $V'$ does not have $H$ as a subgraph), and showed it is hard to approximate within factor $2^{\log^{1/2 - \epsilon} n}$ for any $\epsilon > 0$. They also mentioned the minimization version of two extensions of $H$-Transversal. The most general node-deletion problem is APX-hard for every nontrivial hereditary (i.e. closed under node deletion) property, and the special case where the property is characterized by a finite number of forbidden subgraphs (i.e. $\{H_1, ..., H_l \}$-Transversal in our terminology) can be approximated with a constant ratio. They did not provide explicit constants (one trivial approximation ratio for $\{H_1, ..., H_l \}$-Transversal is $\max (|V_{H_1}|, ..., |V_{H_l}|)$), and our result can be viewed as a quantitative extension of their inapproximability results for the special case of $H$-Transversal.

$H$-Transversal / Packing has been also studied outside the approximation algorithms community. 
The duality between our $H$-Transversal and $H$-Packing is closely related to the famous Erd\H{o}s-P\'{o}sa property actively studied in combinatorics.
The recent work of Jansen and Marx~\cite{JD15} considered problems similar to our $H$-Packing with respect to fixed-parameter tractability (FPT).

Many other works on $H$-Transversal / Packing focus on a special case where $H$ is a cycle or clique. We define $k$-Cycle (resp. $k$-Clique) to be the cycle (resp. clique) on $k$ vertices.

\subsubsection{Cycles}
The initial motivation for our work was to prove a super-constant factor inapproximability for the Feedback Vertex Set (FVS) problem without relying on the Unique Games Conjecture.
Given a (directed) graph $G$, the FVS problem asks to find a subset $F$ of vertices
with the minimum cardinality that intersects every cycle in the graph (equivalently, the induced subgraph $G \setminus F$ is acyclic). 
One of Karp's 21 NP-complete problems, FVS has been a subject of active research for many years in terms of approximation algorithms and fixed-parameter tractability (FPT). For FPT results, see~\cite{Bodlaender94,CLLOR08,CPPW11,CCHM12} and references therein. 


FVS on undirected graphs has a 2-approximation algorithm~\cite{BBF95,BG96,CGHW98}, but the same problem is not well-understood in directed graphs.
The best approximation algorithm~\cite{Seymour95,ENSS98,ENRS00} achieves an approximation factor 
of $O(\log n \log \log n)$.
The best hardness result follows from a simple approximation preserving reduction from Vertex Cover, which implies that it is NP-hard to approximate FVS within a factor of $1.36$~\cite{DS05}. Assuming UGC~\cite{Khot02}, it is NP-hard to approximate FVS in directed graphs within any constant factor~\cite{GMR08,Svensson12} (we give a simpler proof in~\cite{GL14b}). 
The main challenge is to {\em bypass} the UGC and to show a super-constant inapproximability result for FVS assuming only $\mathsf{P} \neq \mathsf{NP}$ or $\mathsf{NP} \not \subseteq \mathsf{BPP}$.

By Theorem~\ref{thm:main_v}, we prove that $k$-Cycle Transversal is hard to approximate within factor $\Omega(k)$.
The following theorem improves the result of Theorem~\ref{thm:main_v} in the sense that in the completeness case, a small number of vertices not only intersect cycles of length exactly $k$, but intersect every cycle of length $3, 4, ..., O(\frac{\log n}{\log \log n})$.

\begin{theorem}
\label{thm:random}
Fix an integer $k \geq 3$ and $\epsilon \in (0,1)$. 
Given a graph $G = (V_G, E_G)$ (directed or undirected), unless $\mathsf{NP} \subseteq \mathsf{BPP}$, there is no polynomial time algorithm to tell apart the following two cases.
\begin{itemize}
\item Completeness: There exists $F \subseteq V_G$ with $\frac{1}{k-1} + \epsilon$ fraction of vertices that intersects every cycle of at most length $O(\frac{\log n}{\log \log n})$ (hidden constant in $O$ depends on $k$ and $\epsilon$). 
\item Soundness: Every subset $F$ with less than $1 - \epsilon$ fraction of vertices does not intersect at least one cycle of length $k$. Equivalently, any subset with more than $\epsilon$ fraction of vertices has a cycle of length exactly $k$ in the induced subgraph. 
\end{itemize}
\end{theorem}

This can be viewed as some (modest) progress towards showing inapproximability of FVS in the following sense. Consider the following standard linear programming (LP) relaxation for FVS.
%
\begin{equation*}
\min \sum_{v \in V_G} x_v  \quad
\mbox{subject to} \quad   \sum_{v \in C} x_v \geq 1 \quad \forall \mbox{ cycle }C \ , \quad \text{and} 
\quad 0 \leq x_v \leq 1 \ \ \forall v \in V_G
\end{equation*}
The integrality gap of the above LP is upper bounded by $O(\log n)$ for undirected graphs~\cite{BGNR98} and $O(\log n \log \log n)$ for directed graphs~\cite{ENSS98}. Suppose in the completeness case, there exists a set of measure $c$ that intersects every cycle of length at most $\log^{1.1} n$ (or any number bigger than the known integrality gaps). If we remove these vertices and consider the above LP on the remaining subgraphs, since every cycle is of length at least $\log^{1.1} n$, setting $x_v = 1 / \log^{1.1} n$ is a feasible solution, implying that the optimal solution to the LP is at most $n / \log^{1.1} n$. Since the integrality gap is at most $O(\log n \log \log n)$, we can conclude that the remaining cycles can be hit by at most $O(n \log \log n / \log^{0.1} n) = o(n)$ vertices, extending the completeness result to every cycle. 
Thus, improving our result to hit cycles of length $\omega(\log n \log \log n)$ in the completeness case will prove a factor-$\omega(1)$ inapproximability of FVS. 

Another interesting aspect about Theorem~\ref{thm:random} is that it also holds for undirected graphs. This should be contrasted with the fact that undirected graphs admit a 2-approximation algorithm for FVS, suggesting that to overcome $\log n$-cycle barrier mentioned above, some properties of directed graphs must be exploited. Towards developing a directed graph specific approach, we also present a different reduction technique called {\em labeling gadget} in Appendix~\ref{subsec:labeling_gadget} to prove a similar result only on directed graphs. It has an additional advantage of being derandomized and assumes only $\mathsf{P} \neq \mathsf{NP}$. 

For cycles of bounded length, Kortsarz et al.~\cite{KLN10} studied $k$-Cycle Edge Transversal, and suggested a $(k - 1)$-approximation algorithm as well as proved that improving the ratio $2$ for $K_3$ will have the same impact on Vertex Cover, refuting the Unique Games Conjecture~\cite{KR08}. 

For the dual problem of packing cycles of any length, called Vertex-Disjoint Cycle Packing (VDCP), the results of~\cite{KNSYY07,FS07} imply that the best approximation factor by any polynomial time algorithm lies between $\Omega(\sqrt{\log n})$ and $O(\log n)$. In a closely related problem Edge-Disjoint Cycle Packing (EDCP), the same papers showed that $\Theta(\log n)$ is the best possible. 
In directed graphs the vertex and edge version have the same approximability, the best known algorithms achieves $O(\sqrt{n})$-approximation while the best hardness result remains $\Omega(\log n)$. 

Variants of $k$-Cycle Packing have also been considered in the literature. 
Rautenbach and Regen~\cite{RR09} studied $k$-Cycle Edge Packing on graphs with girth $k$ and small degree. 
Chalermsook et al.~\cite{CLN14} studied a variant of $k$-Cycle Packing on directed graphs for $k \geq n^{1/2}$ where we want to pack as many disjoint cycles of length at most $k$ as possible, and proved that it is NP-hard to approximate within a factor of $n^{1/2 - \epsilon}$. This matches the algorithm implied by~\cite{KNSYY07}.

\subsubsection{Cliques}
\label{subsubsec:cliques}
Minimum Maximal (resp. Maximum) Clique Transversal asks to find the smallest subset of vertices that intersects every maximal (resp. maximum) clique in the graph. In mathematics, Tuza~\cite{Tuza91} and Erd\H{o}s et al.~\cite{EGT92} started to estimate the size of the smallest such set depending on structure of graphs. See the recent work of Shan et al.~\cite{SLK14} and references therein. 
In computer science, exactly computing the smallest set on special classes of graphs appears in many works \cite{GP01,LCS02,CKL01,DLMS08,Lee12}. 

Both the edge and vertex version of $k$-Clique Packing also have been studied actively both in mathematics and computer science. 
In mathematics, the main focus of research is lower bounding the maximum number of edge or vertex-disjoint copies of $K_k$ in very dense graphs (note that even $K_3$ does not exist in $K_{n,n}$ which has $2n$ vertices and $n^2$ edges). See the recent paper~\cite{Yuster14} or the survey~\cite{Yuster07} of Yuster. The latter survey also mentions approximation algorithms, including APX-hardness and the general approximation algorithm for $k$-Set Packing which now achieves $\frac{k + 1 + \epsilon}{3}$ for the vertex version and $\frac{\binom{k}{2} + 1 + \epsilon}{3}$ for the edge version~\cite{Cygan13}. 
Feder and Subi~\cite{FS12} considered $H$-Edge Packing and showed APX-hardness when $H$ is $k$-cycle or $k$-clique. 
Chataigner et al.~\cite{CMWY09} considered an interesting variant where we want to pack vertex-disjoint cliques of any size to maximize the total number of edges of the packed cliques, and proved APX-hardness and a 2-approximation algorithm.
Exact algorithms for special classes of graphs have been considered in~\cite{BCD97,GPCCW01,HKNP05,Kloks12}.

\subsection{Open Problems}
\label{sec:conclusions}
For $H$-Transversal, 1-connectivity is not sufficient for strong hardness, because $k$-Star Transversal admits an $O(\log k)$-approximation algorithm by Theorem~\ref{thm:star}. It is open whether 1-connectivity is sufficient or not for such strong hardness for $H$-Packing. $k$-Star Packing is at least as hard as MIS-$k$ by a trivial reduction, but the approximability of $k$-Path Packing appears to be still unknown. Whether $k$-Path Transversal admits a factor $o(k)$ approximation algorithm is also an intriguing question. 
For directed acyclic graphs, Svensson~\cite{Svensson12} proved that it is Unique Games-hard to approximate $k$-Path Transversal within a factor better than $k$.

The approximability of $H$-Edge Transversal and $H$-Edge Packing is less understood than the vertex versions. Proving tight characterizations for the edge versions similar to Theorem~\ref{thm:main_v} is an interesting open problem.

\subsection{Organization}
The rest of the main body is devoted to proving Theorem~\ref{thm:main_v} for $H$-Transversal / Packing and Theorem~\ref{thm:mis} for MIS-$k$-$g$. 
Section~\ref{sec:prelim} recalls and extends previous hardness results for the problems we reduce from; 
Sections~\ref{sec:transversal} and~\ref{sec:packing_main} prove hardness of $H$-Transversal and $H$-Packing respectively. 
Appendix~\ref{sec:star} gives an $O(\log k$)-approximation algorithm for $k$-Star Transversal, proving Theorem~\ref{thm:star}.
Appendix~\ref{sec:cycles} proves Theorem~\ref{thm:random} to illustrate the connection to FVS.

\section{Preliminary}
\label{sec:prelim}
\noindent {\bf Notation.} 
A $k$-uniform hypergraph is denoted by $P = (V_P, E_P)$ such that each $e \in E_P$ is a $k$-subset of $V_P$. 
We denote $e$ as an {\em ordered} $k$-tuple $e = (v^1, \dots, v^k)$. The ordering can be chosen arbitrarily given $P$, but should be fixed throughout. 
If $v$ indicates a vertex of some graph, we use a superscript $v^i$ to denote another vertex of the same graph, and $e^i$ to denote the $i$th (hyper)edge.
For an integer $m$, let $[m] = \{ 1, 2, \dots, m \}$. 
Unless otherwise stated, the {\em measure} of $F \subseteq V$ is obtained under the uniform measure on $V$, which is simply $\frac{|F|}{|V|}$. 
%

\medskip
\noindent {\bf $k$-HVC.}
An instance of $k$-HVC consists of a $k$-uniform hypergraph $P$, where the goal is to find a set $C \subseteq V_P$ with the minimum cardinality such that  it intersects every hyperedge. The result of Dinur, Guruswami, Khot and Regev~\cite{DGKR05} states that

\begin{theorem} [\cite{DGKR05}]
\label{thm:dgkr_body}
Given a $k$-uniform hypergraph ($k \geq 3$) and $\epsilon > 0$, it is NP-hard to tell apart the following cases:
\begin{itemize}
\item Completeness: There exists a vertex cover of measure $\frac{1 + \epsilon}{k - 1}$.
\item Soundness: Every vertex cover has measure at least $1 - \epsilon$.
\end{itemize}
Therefore, it is NP-hard to approximate $k$-HVC within a factor $k - 1 + 2\epsilon$. 
\end{theorem}

Moreover, the above result holds even when the degree of a hypergraph is bounded by $d$ depending on $k$ and $\epsilon$. See Appendix~\ref{subsec:hvc} for details. 

\medskip
\noindent {\bf MIS-$k$.}
Given a graph $G = (V_G, E_G)$, a subset $S \subseteq V_G$ is {\em independent} if the subgraph induced by $S$ does not contain any edge. 
The Maximum Independent Set (MIS) problem asks to find the largest independent set, and MIS-$k$ indicates the same problem where $G$ is promised to have maximum degree at most $k$. The recent result of Chan~\cite{Chan13} implies
\begin{theorem}[\cite{Chan13}]
\label{thm:chan}
Given a graph $G$ with maximum degree at most $k$, it is NP-hard to tell apart the following cases:
\begin{itemize}
\item Completeness: There exists an independent set of measure $\Omega(1/(\log k))$.
\item Soundness: Every subset of vertices of measure $O(\frac{\log^3 k}{k})$ contains an edge. 
\end{itemize}
Therefore, it is NP-hard to approximate MIS-$k$ within a factor $\Omega(\frac{k}{\log^4 k})$. 
\end{theorem}

\section{$H$-Transversal}
\label{sec:transversal}
In this section, given a 2-connected graph $H = (V_H, E_H)$ with $k$ vertices, we give a reduction from $k$-HVC to $H$-Transversal. 
The simplest try will be, given a hypergraph $P = (V_P, E_P)$ (let $n = |V_P|, m = |E_P|$), to produce a graph $G = (V_G, E_G)$ where $V_G = V_P$, and for each hyperedge $e = ( v^1, \ldots , v^k )$ add $|E_H|$ edges that form a {\em canonical copy} of $H$ to $E_G$. While the soundness follows directly (if $F \subseteq V_P$ contains a hyperedge, the subgraph induced by $F$ contains $H$), the completeness property does not hold since edges that belong to different canonical copies may form an unintended non-canonical copy. To prevent this, a natural strategy is to replace each vertex by a set of many vertices (call it a {\em cloud}), and for each hyperedge $( v^1, \ldots , v^k )$, add many canonical copies on the $k$ clouds (each copy consists of one vertex from each cloud). If we have too many canonical copies, soundness works easily but completeness is hard to show due to the risk posed by non-canonical copies, and in the other extreme, having too few canonical copies could result in the violation of the soundness property. Therefore, it is important to control the structure (number) of canonical copies that ensure both completeness and soundness at the same time.

Our technique, which we call {\em random matching}, proceeds by creating a carefully chosen number of {\em random} copies of $H$ for each hyperedge to ensure both completeness and soundness.
We remark that properties of random matchings are also used to bound the number of short non-canonical paths in 
 inapproximability results for edge-disjoint paths on undirected graphs~\cite{AZ-edp,ACGKTZ10}. The details in our case are different as we create {\em many} copies of $H$ based on a {\em hypergraph}.

Fix $\epsilon > 0$, apply Theorem~\ref{thm:dgkr_body}, let $c := \frac{1 + \epsilon}{k - 1} , s := 1 - \epsilon$ be the measure of the minimum vertex cover in the completeness and soundness case respectively, and $d := d(k, \epsilon)$ be the maximum degree of hard instances. 
Let $a$ and $B$ be integer constants greater than 1, which will be determined later. 
Lemma~\ref{lem:random_completeness} and~\ref{lem:random_soundness} with these parameters imply the first half of Theorem~\ref{thm:main_v}.

\medskip\noindent {\bf Reduction.} 
Without loss of generality, assume that $V_H = [k]$. 
Given a hypergraph $P = (V_P, E_P)$, construct an undirected graph $G = (V_G, E_G)$ such that
\begin{itemize}
\item $V_G = V_P \times [B]$. Let $n = |V_P|$ and $N = |V_G| = nB$. For $v \in V_P$, let $\mathsf{\mathsf{cloud}}(v) := \left\{ v \right\}  \times [B]$ be the copy of $[B]$ associated with $v$. 
\item For each hyperedge $e = (v^1, \ldots, v^k)$, for $a B$ times, take $l^1, \ldots, l^k$ independently and uniformly from $[B]$. 
For each edge $(i, j) \in H$ $(1 \leq i < j \leq k)$, add $((v^i, l^i), (v^{j}, l^{j}))$ to $E_G$. Each time we add $|E_H|$ edges isomorphic to $H$, and we have $aB$ of such copies of $H$ per each hyperedge. Call such copies {\em canonical}.
\end{itemize}

\medskip\noindent {\bf Completeness.} 
The next lemma shows that if $P$ has a small vertex cover, $G$ also has a small $H$-Transversal.
\begin{lemma}
Suppose $P$ has a vertex cover $C$ of measure $c$. For any $\epsilon > 0$, with probability at least $3/4$, there exists a subset $F \subseteq V_G$ of measure at most $c + \epsilon$ such that the subgraph induced by $V_G \setminus F$ has no copy of $H$. 
\label{lem:random_completeness}
\end{lemma}

\begin{proof}
Let $F = C \times [B]$.
We consider the expected number of copies of $H$ that avoid $F$ and argue that a small fraction of additional vertices intersect all of these copies.
Choose $k$ vertices $(v^1, l^1), \ldots, (v^{k}, l^{k})$ which satisfy
\begin{itemize}
\item $v^1 \in V_P$ can be any vertex.
\item $l^1, \ldots, l^{k} \in B$ can be arbitrary labels. 
\item For each $(i, j) \in E_H$, there must be a hyperedge of $P$ containing both $i$ and $j$. 
\end{itemize}
There are $n$ possible choices for $v^1$, $B$ choices for each $l^i$, and at most $kd$ choices for each $v^i \, (i > 1)$. 
The number of possibilities to choose such $(v^1, l^1), \ldots, (v^{k}, l^{k})$ is bounded by $n (dk)^{k} B^{k}$. 
Note that no other $k$-tuple of vertices induce a connected graph and contain a copy of $H$. Further discard the tuple when two vertices are the same. 

We calculate the probability that the subgraph induced by $((v^1, l^1), \ldots, (v^{k}, l^{k}))$ contains a copy {\em in this order} --- formally, for all $(i, j) \in E_H$, $((v^i, l^i), (v^j, l^j)) \in E_G$. 
For each $(i, j) \in E_H$, we call a pair $((v^i, l^i), (v^j, l^j)) \in \binom{V_G}{2}$ a {\em purported edge}. 
For a set of purported edges, we say that this set can be {\em covered by a single canonical copy} if one copy of canonical copy of $H$ can contain all purported edges with nonzero probability. 
Suppose that all $|E_H|$ purported edges can be covered by a single canonical copy of $H$. 
It is only possible when there is a hyperedge whose $k$ vertices are exactly $\{ v^1, \ldots, v^{k} \}$. 
In this case, $((v^1, l^1), \ldots, (v^{k}, l^{k}))$ intersects $F$. (right case of Figure~\ref{fig:partition}). 
When $|E_H|$ purported edges have to be covered by more than one canonical copy, some vertices must be covered by more than one canonical copy, and each canonical copy covering the same vertex should give the same label to that vertex. This redundancy makes it unlikely to have all $k$ edges exist at the same time. (left case of Figure~\ref{fig:partition}). 
The below claim formalizes this intuition.

\begin{figure}
  \centering
      \includegraphics[height=7cm]{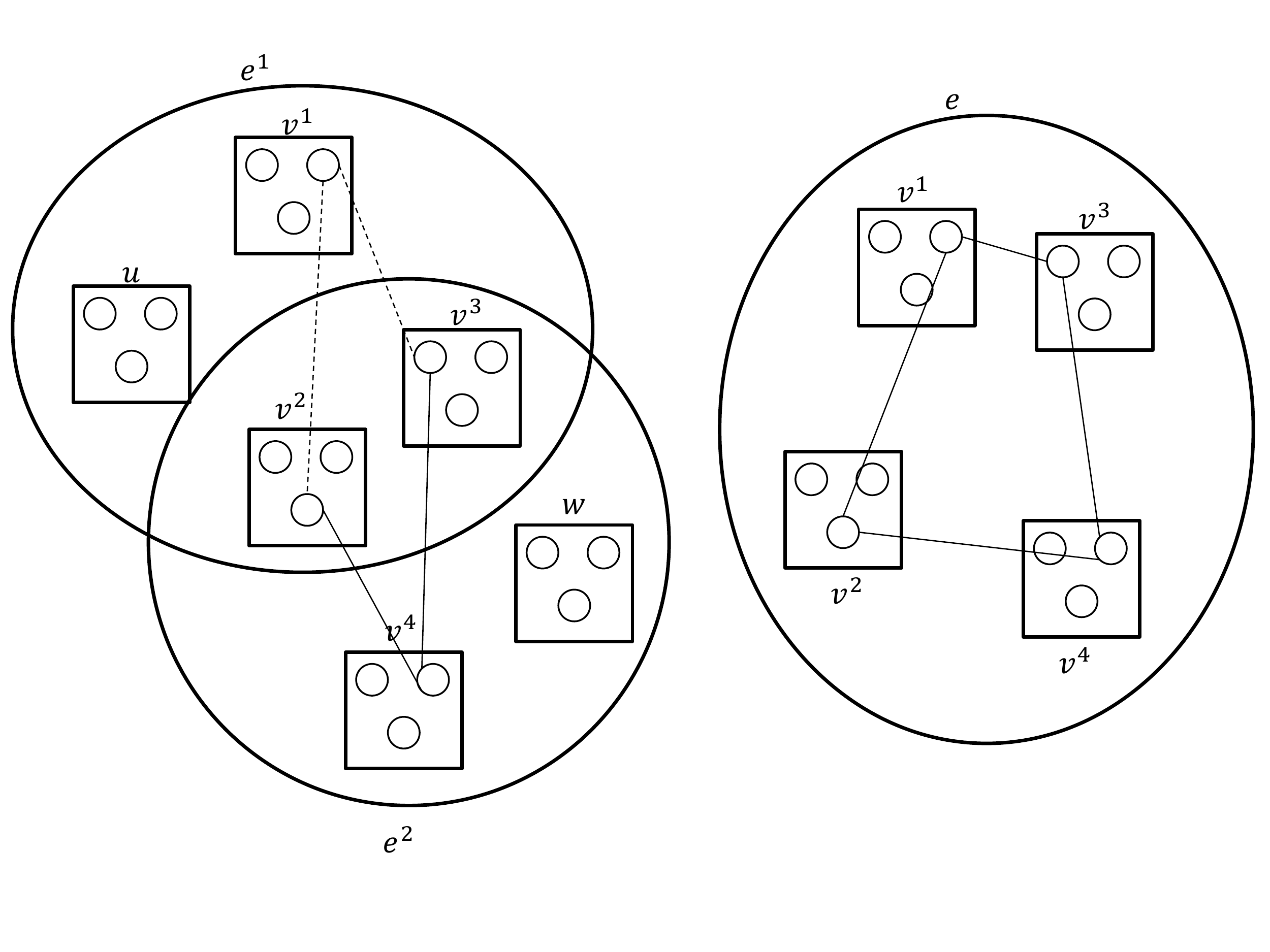}
  \caption{{\small Two examples where $k = 4$ and $H$ is a 4-cycle. On the left, purported edges are divided into two groups (dashed and solid edges). Each copy of canonical cycle should match the labels of three vertices to ensure it covers 2 designated edges (6 labels total). On the right, one canonical copy can cover all the edges, and it only needs to match the labels of four vertices (4 labels total).}
}
  \label{fig:partition}
\end{figure}

\begin{claim}
Suppose that $((v^1, l^1), \ldots, (v^{k}, l^{k}))$ cannot be covered by a single canonical copy. Then the probability that it forms a copy of $H$ is at most $\frac{(adk)^{k^2}}{B^k}$.
\label{claim:numberofcycles}
\end{claim}

\begin{proof}
Fix $2 \leq p \leq |E_H|$. Partition $|E_H|$ purported edges into $p$ nonempty {\em groups} $I_1, \ldots, I_p$
such that each group can be covered by a single canonical copy of $H$. 
There are at most $p^{|E_H|}$ possibilities to partition.
For each $v \in V_P$, there are at most $d$ hyperedges containing $v$ and at most $aBd$ canonical copies intersecting $\mathsf{\mathsf{cloud}}(v)$. 
Therefore, all edges in one group can be covered simultaneously by at most $aBd$ copies of canonical copies. 
There are at most $(aBd)^p$ possibilities to assign a canonical copy to each group. 
Assume that one canonical copy is responsible for exactly one group. 
This is without loss of generality since if one canonical copy is responsible for many groups, we can merge them and this case can be dealt with smaller $p$. 

Focus on one group $I$ of purported edges, and one canonical copy $L = (V_L, E_L)$ which is supposed to cover them. Let $I' \subseteq V_G$ be the set of vertices which are incident on the edges in $I$. 
Suppose $V_L = \{ (u^1, l'^1), \ldots, (u^k, l'^k) \}$, which is created by a hyperedge $f = (u^1, \ldots, u^k) \in E_P$. 
We calculate the probability that $L$ contains all edges in $I$ over the choice of labels $l'^1, \ldots, l'^k$ for $L$. 
One necessary condition is that 
$\left\{ v | (v, l) \in I' \mbox{ for some } l \in [B] \right\}$ (i.e. the set $I'$ projected to $V_P$) is contained in $f$. Otherwise, some vertices of $I'$ cannot be covered by $L$. 
Another necessary condition is $v^i \neq v^j$ for any $(v^i, l^i) \neq (v^j, l^j) \in I'$.
Otherwise (i.e. $(v, l^i), (v, l^j) \in I'$ for $l^i \neq l^j$), since $L$ gives only one label to each vertex in $f \subseteq V_P$, $(v, l^i)$ and $(v, l^j)$ cannot be contained in $L$ simultaneously. Therefore, we have a nice characterization of $I'$: It consists of at most one vertex from the cloud of each vertex in $f$.

The probability that $L$ contains $I$ is at most the probability that for each $(v^i, l^i) \in I'$, $l^i$ is equal to the label $L$ assigns to $v^i$, which is $B^{-|I'|}$. Now we need the following lemma saying that the sum of $|I'|$ is large, which relies on 2-connectivity of $H$.

\begin{lemma}
Fix $p \geq 2$. For any partition $I_1, ..., I_p$ of purported edges into $p$ non-empty groups, $\sum_{i = 1}^p |I'_i| \geq k + p$.
\end{lemma}
\begin{proof}
Let $t$ be the number of vertices contained in at least two $I'_i$s. 
Call them {\em boundary vertices}. 
Note that exactly $k - t$ vertices belongs to exactly one $I'_i$. 
For $i = 1, ..., p$, let $b_i$ be the number of boundary vertices in $|I'_i|$. Since $(I'_i, I_i)$ is a proper subgraph of $H$ and $H$ is 2-vertex connected, $b_i \geq 2$ for each $i$. Therefore, 
\[
\sum_{i=1}^{p} |I'_i| = (k - t) + \max(2p, 2t) \geq k + p.
\]
\end{proof}


We conclude that for each partition, the probability of having all the edges is at most 
\[ (aBd)^p \prod_{q= 1}^p{B^{-|I'_q|}}=  
\frac{(aBd)^p}{B^{k + p}} = \frac{(ad)^p}{B^{k}} \ .\]
The probability that $((v^1, l^1), \ldots, (v^{k}, l^{k}))$ forms a copy is therefore bounded by 
\[  \sum_{p=2}^{|E_H|} \ p^{|E_H|} \frac{(ad)^p}{B^{k}} \leq \frac{(adk)^{k^2}}{B^k} \ . \]
\end{proof}

Therefore, the expected number of copies that avoid $F$ is bounded by
$n(kd)^{k}B^{k} \cdot \frac{(adk)^{k^2}}{B^k}$. With probability at least $3/4$, the number of such copies is at most $4 n (adk)^{2k^2}$.
Let $B \geq \frac{4 (adk)^{2k^2}}{\epsilon}$. Then these copies of $H$ can be covered by at most $\epsilon n B = \epsilon N$ vertices. 
\end{proof}

\medskip\noindent {\bf Soundness.} 
The soundness claim above is easier to establish. By an averaging argument, a subset $I$ of $V_G$ of measure $2 \epsilon$ must contain $\epsilon B$ vertices from the clouds corresponding to a subset $S$ of measure $\epsilon$ in $V_P$. There must be a hyperedge $e$ contained within $S$, and the chosen parameters ensure that one of the canonical copies corresponding to $e$ is likely to lie within $I$.

\begin{lemma}
For $a = a(k, \epsilon)$ and $B = \Omega(\log |E_P|)$, 
if every subset of $V_P$ of measure at least $\epsilon$ contains a hyperedge in the induced subgraph, with probability at least $3/4$, every subset of $V_G$ with measure $2\epsilon$ contains a canonical copy of $H$. 
\label{lem:random_soundness}
\end{lemma}
\begin{proof}
We want to show that the following property holds for every hyperedge $e = (v^1, \dots, v^k)$:  
if a subset of vertices $I \subseteq V_G$ has at least $\epsilon$ fraction of vertices from each $\mathsf{\mathsf{cloud}}(v^i)$, then $I$ will contain a canonical copy. 
Fix $A^1 \subseteq \mathsf{\mathsf{cloud}}(v^1), \dots , A^k \subseteq \mathsf{\mathsf{cloud}} (v^k)$ be such that for each $i$, $|A^i| \geq \epsilon B$.
There are at most $2^{kB}$ ways to choose such $A$'s. 
The probability that one canonical copy associated with $e$ is not contained in $(v^1, A^1) \times \cdots \times (v^k, A^k)$ is at most $1 - \epsilon^k$. 
The probability that none of canonical copy associated with $e$ is contained in $(v^1, A^1) \times \cdots \times (v^k, A^k)$ is 
$(1 - \epsilon^k)^{aB} \leq \exp(-aB\epsilon^k)$. 

By union bound over all $A^1, \ldots, A^k$, the probability that there exists $A^1, \ldots, A^k$ containing no canonical copy is at most $\exp(kB - aB \epsilon^k) = \exp(-B) \leq \frac{1}{4|E_P|}$ by taking $a$ large enough constant depending on $k$ and $\epsilon$, and $B = \Omega( \log |E_P|)$. 
Therefore, with probability at least $3/4$, the desired property holds for all hyperedges. 

Let $I$ be a subset of $V_G$ of measure at least $2\epsilon$. By an averaging argument, at least $\epsilon$ fraction of {\em good} vertices $v \in V_P$ satisfy that $|\mathsf{\mathsf{cloud}}(v^i) \cap I| \geq \epsilon B$. 
By the soundness property of $P$, there is a hyperedge $e$ contained in the subgraph induced by the good vertices, and the above property for $e$ ensures that $I$ contains a canonical copy.
\end{proof}

\section{$H$-Packing and MIS-$k$-$g$}
\label{sec:packing_main}
Given a 2-connected graph $H$, the reduction from MIS-$k$-$k$ to $H$-Packing is relatively straightforward. 
Here we assume that hard instances of MIS-$k$-$k$ are indeed $k$-regular for simplicity. 
Given an instance $M = (V_M, E_M)$ of MIS-$k$-$k$, we take $G = (V_G, E_G)$ to be its {\em line graph} --- $V_G = E_M$, and $e, f \in V_G$ are adjacent if and only if they share an endpoint as edges of $M$. 

For each vertex $v \in V_M$, let $\Star(v) := \{ e \in V_G : v \in e \}$. $\Star(v)$ induces a $k$-clique, and for $v, u \in V_M$, $\Star(v)$ and $\Star(u)$ share one vertex if $u$ and $v$ are adjacent, and share no vertex otherwise. Given an independent set $S$ of $M$, we can find $|S|$ pairwise disjoint stars in $G$, which gives $|S|$ vertex-disjoint copies of $H$. On the other hand, 2-connectivity of $H$ and large girth of $M$ implies that any copy of $H$ must be entirely contained in one star, proving that many disjoint copies of $H$ in $G$ also give a large independent set of $M$ with the same cardinality, completing the reduction from MIS-$k$-$k$ to $H$-Packing.
The following theorem formalizes the above intuition.


\begin{lemma}
\label{lem:packing_reduction}
For a 2-connected graph $H$ with $k$ vertices, there is an approximation-preserving reduction from MIS-$k$-$k$ to $H$-Packing.
\end{lemma}
\begin{proof}
Let $M = (V_M, E_M)$ be an instance of MIS-$k$-$k$ $M$ with maximum degree $k$ and girth greater than $k$. 
First, let $G = (V_G = E_M, E_G)$ be the line graph of $M$. For each vertex $v \in V_M$ with degree strictly less than $k$, we add $k - \Deg(v)$ 
new vertices to $V_G$. Let $\Star(v) \subseteq V_G$ be the union of the edges of $M$ incident on $v$ and the newly added vertices for $v$. 
Note that $|\Star(v)| = k$ for all $v \in V_M$. Add edges to $G$ to ensure that every $\Star(v)$ induces a $k$-clique. 
For two vertices $u$ and $v$ of $M$, $\Star(u)$ and $\Star(v)$ share exactly one vertex if $u$ and $v$ are adjacent in $M$, and share no vertex otherwise. 

Let $S$ be an independent set of $M$. The $|S|$ stars $\{ \Star(v) \}_{v \in S}$ are pairwise disjoint and each induces a $k$-clique, so $G$ contains at least $|S|$ disjoint copies of $H$. 

We claim that any $k$-subset of $V_G$ that induces a 2-connected subgraph must be $\Star(v)$ for some $v$. 
Assume towards contradiction, let $T$ be a $k$-subset inducing a 2-connected subgraph of $G$ that cannot be contained in a single star.
We first show $T$ must contain two disjoint edges of $M$. Take any $(u, v) \in T$. Since $T \notin \Star(u)$, $T$ contains an edge of $M$ not incident on $u$. If it is not incident on $v$ either, we are done. Otherwise, let $(w, v)$ be this edge. The same argument from $T \notin \Star(v)$ gives another edge $(w', u)$ in $T$. If $w \neq w'$, $(w, v)$ and $(w', u)$ are disjoint. Otherwise, $w, u, v$ form a triangle in $M$, contradicting a large girth.
Let $(u, v)$, $(w, x)$ be two disjoint edges of $M$ in contained in $T$. 

Since the subgraph of $G$ vertex-induced by $T$ is 2-connected, 
there are two internally vertex-disjoint paths $P_1$, $P_2$ in $G$ from $(u, v)$ to $(w, x)$. 
The sum of the two lengths is at most $k$, where the length of a path is defined to be the number of edges.
By considering the internal vertices of $P_i$ (edges of $M$) and deleting unnecessary portions, 
we have two edge-disjoint paths $P'_1$, $P'_2$ in $M$ where each $P'_i$ connects $\{ u, v \}$ and $\{ w, x \}$, with length at most the length of $P_i$ minus one. 
There is a cycle in $M$ consists only of the edges of $P'_1$, $P'_2$ together with $(u, v), (w, x)$. 
Since $|P'_1| + |P'_2| + 2 \leq k$, it contradicts that $M$ has girth strictly greater than $k$. 
\end{proof}

We prove that MIS-$k$-$g$ is also hard to approximate by a reduction from MIS-$d$ ($d = \widetilde\Omega(k)$),
using a slightly different {\em random matching} idea.
Given a degree-$d$ graph with possibly small girth, 
we replace each vertex by a cloud of $B$ vertices, 
and replace each edge by $a$ copies of random matching between the two clouds. 
While maintaining the soundness guarantee, we show that there are only a few small cycles, and by deleting a vertex from each of them and {\em sparsifying} the graph we obtain a hard instance for MIS-$k$-$g$.
Note that $g$ does not affect the inapproximability factor but only the runtime of the reduction.

\begin{theorem}
[Restatement of Theorem~\ref{thm:mis}]
For any constants $k$ and $g$, unless $\mathsf{NP} \subseteq \mathsf{BPP}$, no polynomial time algorithm approximates MIS-$k$-$g$ within a factor of $\Omega(\frac{k}{\log^7 k})$.
\end{theorem}
\begin{proof}
We reduce from MIS-$d$ to MIS-$k$-$g$ where $k = O(d \log^2 d)$. 
Given an instance $G_0 = (V_{G_0}, E_{G_0})$ of MIS-$d$, we construct $G = (V_G, E_G)$ and $G' = (V_{G'}, E_{G'})$ by the following procedure:
\begin{itemize}
\item $V_G = V_{G_0} \times [B]$. As usual, let $\mathsf{cloud}(v) = \left\{ v \right\} \times [B]$.
\item For each edge $(u, v) \in E_{G_0}$, for $a$ times, add a {\em random matching} as follows.
\begin{itemize}
\item Take a random permutation $\pi : [B] \rightarrow [B]$.
\item Add an edge $((u, i), (v, \pi(i))$ for all $i \in [B]$.
\end{itemize}
\item Call the resulting graph $G$. To get the final graph $G'$, 
\begin{itemize}
\item For any cycle of length at most $g$, delete an arbitrary vertex from the cycle. Repeat until there is no cycle of length at most $g$.
\end{itemize}
\end{itemize}
Note that the step of eliminating the small cycles can be implemented trivially in time $O(n^g)$.
Let $n = |V_{G_0}|, m = |E_{G_0}|, N = nB = |V_{G}| \geq |V_{G'}|, M = m \cdot aB = |E_G| \geq |E_{G'}|$.
The maximum degree of $G$ and $G'$ is at most $ad$. By construction, girth of $G'$ is at least $g + 1$.

\medskip\noindent {\bf Girth Control.}
We calculate the expected number of small cycles in $G$, and argue that the number of these cycles is much smaller than the total number of vertices, so that $|V_G|$ and $|V_{G'}|$ are almost the same. 
Let $k'$ be the length of a purported cycle. Choose $k'$ vertices $(v^1, l^1), \ldots, (v^{k'}, l^{k'})$ which satisfy
\begin{itemize}
\item $v^1 \in V_{G_0}$ can be any vertex.
\item For each $1 \leq i < k'$, $(v^i, v^{i+1}) \in E_{G_0}$. 
\item $l^1, \dots , l^{k'} \in B$ can be arbitrary labels. 
\end{itemize}
There are $n$ possible choices for $v^1$, $B$ choices for each $l^i$, and $d$ choices for each $v^i \, (i > 1)$. 
The number of possibilities to choose such $(v^1, l^1), \dots , (v^{k'}, l^{k'})$ is bounded by $n d^{k'-1} B^{k'}$. 
Without loss of generality, assume that no vertices appear more than once. 

For each edge $e = (u, w) \in G_0$, consider the intersection of the purported cycle $((v^1, l^1), ..., (v^{k'}, l^{k'}))$ and the subgraph induced by $\mathsf{cloud}(u) \cup \mathsf{cloud}(w)$. It is a bipartite graph with the maximum degree 2. Suppose there are $q$ purported edges $e^1, \dots , e^{q}$ (ordered arbitrarily) in this bipartite graph. 
By slightly abusing notation, let $e^i$ also denote the event that $e^i$ exists in $G$. The following claim upper bounds $\Pr[e^i | e^1, \dots , e^{i-1}]$ for each $e^i$.

\begin{claim}
\label{claim:probability}
$\Pr[e^i | e^1, \dots , e^{i-1}] \leq \frac{a}{B - i}$.
\end{claim}
\begin{proof}
There are $a$ random matchings between $\mathsf{cloud}(u)$ and $\mathsf{cloud}(w)$, and
for each $j < i$, there is at least one random matching including $e^j$. 
We fix one random matching and calculate the probability that the random matching contains $e^i$, conditioned on the fact that it already contains some of $e^1, \dots,  e^{i-1}$. 

If there is $e^{j} \, (j < i)$ that shares a vertex with $e^i$, $e^i$ cannot be covered by the same random matching with $e^j$. 
If a random matching covers $p$ of $e^1, \dots ,  e^{i-1}$ which are disjoint from $e^i$, 
the probability that $e^i$ is covered by that random matching is $\frac{1}{B - p}$, and this is maximized when $p = i - 1$.

By a union bound over the $a$ random matchings, $\Pr[e^i | e^1, \dots , e^{i-1}] \leq \frac{a}{B - i}$.
\end{proof}

The probability that all of $e^1, \ldots, e^q$ exist is at most 
\[
\prod_{i=1}^{q} \frac{a}{B - i} \leq \left(\frac{a}{B - q}\right)^q \leq \left(\frac{a}{B - k'}\right)^q \ .
\]
Since edges of $G_0$ are processed independently, 
the probability of success for one fixed purported cycle is $(\frac{a}{B- k'})^{k'}$.
The expected number of cycles of length $k'$ is 
\begin{align*}
& n d^{k'-1} B^{k'} \cdot  \Bigl(\frac{a}{B - k'}\Bigr)^{k'} = n d^{k'-1} a^{k'} \biggl( 1 + \frac{k'}{B-k'} \biggr)^{k'}  \\
\leq & n d^{k'-1} a^{k'} \exp \Bigl(\frac{k'^2}{B - k'}\Bigr)
\leq e n(ad)^{k'}
\end{align*}
by taking $B - k' \geq k'^2$. 
Summing over $k' = 1, \dots , g$, the expected number of cycles of length up to $g$, is bounded by $eg(ad)^{g}n$. Take $B \geq 4d^2 \cdot eg(ad)^{g}$. Then with probability at least $3/4$, 
the number of cycles of length at most $g$ is at most $\frac{Bn}{d^2}$. 
By taking $1/d^2$ fraction of vertices away (one for each short cycle), we have a girth at least $g + 1$, which implies
$
\Bigl(1 - \frac{1}{d^2}\Bigr) |V_{G}| \leq  |V_{G'}| \leq |V_{G}|.
$

Hardness of MIS-$d$ states that it is NP-hard to distinguish the case $G_0$ has an independent set of measure $c := \Omega(\frac{1}{\log d})$ and the case where the maximum independent set has measure at most $s := O(\frac{\log^3 d}{d})$. 

\medskip\noindent {\bf Completeness.}
Let $I_0$ be an independent set of $G_0$ of measure $c$. Then $I = I_0 \times [B]$ is also an independent set of $G$ of measure $c$. Let $I' = I \cap V_{G'}$. $I'$ is independent in both $G$ and $G'$, and the measure of $I'$ in $G'$ is at least the measure of $I'$ in $G$, which is at least 
$c - 1/d^2 = \Omega(\frac{1}{\log d})$.

\medskip\noindent {\bf Soundness.}
Suppose that every subset of $V_{G_0}$ of measure at least $s$ contains an edge. 
Say a graph is  \emph{$(\beta, \alpha)$-dense} if we take $\beta$ fraction of vertices, at least $\alpha$ fraction of edges lie within the induced subgraph.
We also say a bipartite graph is \emph{$(\beta, \alpha)$-bipartite dense} if we take $\beta$ fraction of vertices from each side, at least $\alpha$ fraction of edges lie within the induced subgraph.
\begin{claim}
For $a = O(\frac{\log(1/s)}{s})$ and $B=O(\frac{\log m}{s})$ the following holds with probability at least $3/4$: For every  $(u, w) \in E_{G_0}$, the bipartite graph between $\mathsf{cloud}(u)$ and $\mathsf{cloud}(w)$ is $(\epsilon, \epsilon^2 / 8)$-bipartite dense for all $\epsilon \geq s$. 
\label{claim:disperser2}
\end{claim}
\begin{proof}

Fix $(u, w)$, and $\epsilon \in [s, 1]$,
and $X \subseteq \mathsf{cloud}(u)$ and $Y \subseteq \mathsf{cloud}(w)$ be such that
$|X| = |Y| = \epsilon B$. 
The possibilities of choosing $X$ and $Y$ is 
\[
{{B}\choose{\epsilon B}}^2 \leq \exp(O(\epsilon \log (1/\epsilon) B))
\]

Without loss of generality, let $X = Y = [\epsilon B]$. 
In one random matching, let $X_i \, (i \in [\epsilon B])$ be the random variable indicating whether vertex $(u, i) \in X$ is matched with a vertex in $Y$ or not. $\Pr[X_1 = 1] = \epsilon$, and $\Pr[X_i = 1| X_1, \ldots, X_{i-1}] \geq \epsilon / 2$ for $i \in [\epsilon B / 2]$ and any $X_1, \ldots, X_{i-1}$. Therefore, the expected number of edges between $X$ and $Y$ is at least $\epsilon^2 B / 4$.
With $a$ random matchings, the expected number is at least $a \epsilon^2 B / 4$. By Chernoff bound, 
the probability that it is less than $a \epsilon^2 B / 8$ is at most
$\exp(\frac{a \epsilon^2 B}{32})$. By union bound over all possibilities of choosing $X$ and $Y$, the probability that the bipartite graph is not $(\epsilon, \epsilon^2 / 8)$-bipartite dense is 
\[
\exp(\epsilon \log (1/\epsilon) B) \cdot \exp\Bigl( -\frac{a \epsilon^2 B}{32}\Bigr) \leq \frac{1}{4mB}
\]
by taking $a = O(\frac{\log (1/s)}{s})$ and $B = O\bigl(\frac{\log m}{s}\bigr)$.
A union bound over all possible choices of $\epsilon$ ($B$ possibilities) and $m$ edges of $E_0$ implies the claim.
\end{proof}

\begin{claim}
With the parameters $a$ and $B$ above, $G$ is $(4s \log (1/s), \Omega(\frac{s}{d}))$-dense. 
\label{claim:density}
\end{claim}
\begin{proof}
Fix a subset $S$ of measure $4 s \log (1/s)$. 
For a vertex $v$ of $G_0$, let $\mu(v) := \frac{|\mathsf{cloud}(v) \cap S|}{B}$. 
Note that $\mathbb{E}_v [\mu(v)] = 4s \log (1/s)$.
Partition $V_{G_0}$ into $t+1$ buckets $B_0, \ldots, B_{t}$ ($t := \lceil \log_2 (1/ s) \rceil $), 
such that $B_0$ contains $v$ such that $\mu(v) \leq s$, and for $i \ge 1$,
$B_i$ contains $v$ such that $\mu(v) \in (2^{i-1} s, 2^{i} s]$. 
Denote
\[ \mu(B_i) := \frac{\sum_{v \in B_i} \mu(v)}{|V_{G_0}|} \ . \]
Clearly $\mu(B_0) \leq s$. Pick $i \in \{1, \ldots, t \}$ with the largest $\mu(B_i)$.
We have $\mu(B_i) \geq 2 s$ since $\E_v [ \mu(v)] \ge 4 s \log (1/s)$. Let $\gamma = 2^{i-1} s$. 
All vertices of $B_i$ has $\mu(v) \in [\gamma, 2\gamma]$, so 
$|B_i| \geq (s / \gamma) n$. 

Since $G_0$ has no independent set with more than $ns$ vertices, Tur\'{a}n's Theorem says that 
the subgraph of $G_0$ induced by $B_i$ has at least $\frac{|B_i|}{2}(\frac{|B_i|}{ns} - 1) = \Omega(\frac{s}{\gamma^2}n)$ edges.
This is at least $\Omega(\frac{s}{d \gamma^2})$ fraction of the total number of edges. 

For each of these edges, by Claim~\ref{claim:disperser2}, at least $\gamma^2 / 8$ fraction of the edges from the bipartite graph connecting the clouds of its two endpoints, lie in the subgraph induced by $S$ (since $\gamma \geq s$). 
Overall, we conclude that there are at least $\Omega(\frac{s}{d\gamma^2}) \cdot \frac{\gamma^2}{8} = \Omega(\frac{s}{d})$ fraction of edges inside the subgraph induced by $S$.
\end{proof}

\medskip\noindent {\bf Sparsification.}
Recall that $G'$ is obtained from $G$ by deleting at most $\frac{1}{d^2}$ fraction of vertices to have girth greater than $g$. 
In the completeness case, $G'$ has an independent set of measure at least $c - 1/d^2 = \Omega(\frac{1}{\log d})$. 
In the soundness case, $G$ is $(4s \log (1/s), \Omega(\frac{s}{d}))$-dense, so
$G'$ is $(\beta, \alpha)$-dense where $\beta := \Omega(\frac{\log^4 d}{d}), 
\alpha := \Omega(\frac{\log^3 d}{d^2})$. 
Using density of $G'$, we sparsify $G'$ again --- keep each edge of $G'$ by probability $\frac{kn}{|E_{G'}|}$ so that the expected total number of edges is $kn$.

Fix a subset $S \subseteq V_{G'}$ of measure $\beta$. 
Since there are at least $\alpha$ fraction of edges in the subgraph induced by $S$, 
the expected number of picked edges in this subgraph is at least $\alpha k n$. 
By Chernoff bound, the probability that it is less than $\frac{\alpha k n }{8}$ is at most $\exp(- \frac{\alpha k n }{32})$. 
By union bound over all sets of measure exactly $\beta$ (there are at most ${{n} \choose {n \beta}} \leq \exp(2 \beta \log (1 / \beta) n)$ of them), and over all possible values of $\beta$ (there are at most $n$ possible sizes), the desired property fails with probability at most
\[
n \cdot \max_{\beta \in [\beta_0,1]} \bigl\{
\exp(-\alpha k n/32) \cdot \exp(2 \beta \log (1 / \beta) n) \bigr\} \leq n \cdot e^{-n}
\]
when $k = O(\frac{\beta \log (1 / \beta)}{\alpha}) = O(d \log^2 d)$.
In the last step we remove all the vertices of degree more than $10k$.
Since the expected degree of each vertex is at most $2k$, the expected fraction of deleted vertices is $\exp(-\Omega(k)) \ll \beta$.

Combining all these results, we have a graph with small degree $10k = O(d \log^2 d)$ and girth strictly greater than $g$, 
where it is NP-hard to approximate MIS within a factor of $\frac{c - \frac{1}{d^2}}{\beta} = \Omega(\frac{d}{\log^5 d}) = \Omega(\frac{k}{\log^7 k})$. 
Therefore, it is NP-hard to approximate MIS-$k$-$g$ within a factor of $\Omega(\frac{k}{\log^7 k})$.
\end{proof}

\bibliographystyle{abbrv}
\bibliography{mybib}

\parskip=1ex
\appendix
\parskip=1ex

\section{Approximation Algorithm for $k$-Star Transversal}
\label{sec:star}
In this section, we show that $k$-Star Transversal admits an $O(\log k)$-approximation algorithm, matching the $\Omega(\log k)$-hardness obtained via a simple reduction from Minimum Dominating Set on degree-$(k - 1)$ graphs~\cite{CC08}, and proving Theorem~\ref{thm:star}.
Let $G = (V_G, E_G)$ be the instance of $k$-Star Transversal. 
This problem has a natural interpretation that it is equivalent to finding the smallest $F \subseteq V_G$ such that the subgraph induced by $V_G \setminus F$ has maximum degree at most $k - 2$. 
Our algorithm consists of two phases. 
\begin{enumerate}
\item Iteratively solve 2-rounds of Sherali-Adams linear programming (LP) hierarchy and put vertices with a large fractional value in the transversal. If this phase terminates with a partial transversal $F$, the remaining subgraph induced by $V_G \setminus F$ has small degree (at most $2k$) and the LP solution to the last iteration is {\em highly fractional.} 
\item We reduce the remaining problem to {\em Constrained Set Multicover} and use the standard greedy algorithm. 
While the analysis of the greedy algorithm for Constrained Set Multicover is used as a black-box, 
low degree of the remaining graph and high fractionality of the LP solution imply that the analysis is almost tight for our problem as well. 
\end{enumerate}

\medskip \noindent {\bf Iterative Sherali-Adams.}
Given $G$, 2-rounds of Sherali-Adams hierarchy of LP relaxation has variables $\{ x_{v} \}_{v \in V_G} \cup \{ x_{u, v} \}_{u, v \in V_G}$. An integral solution $y : V_G \mapsto \{ 0, 1 \}$, where $y(v) = 1$ indicates that $v$ is picked in the transvesal, naturally gives a feasible solution to the hierarchy by $x_v = y_v$, $x_{u, v} = y_u y_v$. Consider the following relaxation for $k$-Star Transversal. 

\begin{align*}
\mbox{minimize   } \qquad & \sum_{v \in V_G} x_v & \\
\mbox{subject to   }\qquad  
& 0 \leq x_{u, v}, x_v \leq 1 & \forall u, v \in V_G \\
& x_{u, v} \leq x_u & \forall u, v \in V_G \\
& x_{u} + x_v - x_{u, v} \leq 1 & \forall u, v \in V_G \\
& \sum_{v : (u, v) \in E_G} (x_v - x_{u, v}) \geq (\Deg(u) - k + 2)(1 - x_u) & \forall u \in V_G \\
\end{align*}

The first three constraints are common to any 2-rounds of Sherali-Adams hierarchy, and ensure that for any $u, v \in V_G$, the {\em local distribution} on four assignments $\alpha : \{ u, v \} \mapsto \{ 0, 1 \}$ forms a valid distribution. In other words, the following four numbers are nonnegative and sum to 1: $\Pr[\alpha(u) = \alpha(v) = 1] := x_{u, v}$, $\Pr[\alpha(u) = 0, \alpha(v) = 1] := x_v - x_{u, v}$, $\Pr[\alpha(u) = 1, \alpha(v) = 0] := x_u - x_{u, v}$, $\Pr[\alpha(u) = \alpha(v) = 0] := 1 - x_u - x_v + x_{u, v}$. 

The last constraint is specific to $k$-Star Transversal, and it is easy to see that it is a valid relaxation: Given a feasible integral solution $y : V_G \mapsto \{ 0, 1 \}$, the last constraint is vacuously satisfied when $y_u = x_u = 1$, and if not, it requires that at least $\deg(u) - k + 2$ vertices should be picked in the transversal so that there is no copy of $k$-Star in the induced subgraph centered on $u$. 
The first phase proceeds as the following.

\begin{itemize}
\item Let $S \leftarrow \emptyset$. 
\item Repeat the following until the size of $S$ does not increase in one iteration.
\begin{itemize}
\item Solve the above Sherali-Adams hierarchy for $V_G \setminus S$ --- it means to solve the above LP with additional constraints $x_v = 1$ for all $v \in S$, which also implies $x_{u, v} = x_u$ for $v \in S, u \in V_G$. Denote this LP by $\SA(S)$. 
\item $S \leftarrow \{ v : x_v \geq \frac{1}{\alpha} \}$, where $\alpha := 10$.
\end{itemize}
\end{itemize}

We need to establish three properties from the first phase: 
\begin{itemize}
\item The size of $S$ is close to that of the optimal $k$-Star Transversal. 
\item Maximum degree of the subgraph induced by $V_G \setminus S$ is small.
\item The remaining solution has small fractional values --- $x_v < \frac{1}{\alpha}$ for all $v \in V_G \setminus S$. 
\end{itemize}

The final property is satisfied by the procedure. The following two lemmas establish the other two properties. 

\begin{lemma}
Let $\Frac$ be the optimal value of  $\SA(\emptyset)$.
When the above procedure terminates, $|S| \leq \alpha \Frac$. 
\end{lemma}
\begin{proof}
Assume that the above loop iterated $l$ times, and for $i = 0, ..., l$, let $S_i$ be $S$ after the $i$th loop such that $S_0 = \emptyset, ..., S_l = S$. 
We use induction from the last iteration. 
Let $\Frac_i$ be the optimal fractional solution to $\SA(S_i)$ minus $|S_i|$ such that $\Frac = \Frac_0$.

We first establish $|S_l| - |S_{l - 1}| \leq \alpha \Frac_{l - 1}$. This is easy to see because, when $x$ is the optimal fraction solution to $\SA(S_{l - 1})$, 
\[
|S_l| - |S_{l - 1}| = | \{ v \notin S_{l - 1} : x_v \geq \frac{1}{\alpha} \} | 
\leq \alpha \Frac_{l - 1}.
\]

For $i = l - 2, l-1, ..., 0$, we show that $|S_l| - |S_i| \leq \alpha \Frac_i$. 
Let $x$ be the optimal fraction solution to $\SA(S_{i})$, and $x'$ be the solution obtained by {\em partially rounding} $x$ in the following way.
\begin{itemize}
\item $x'_v = 1$ if $v \in S_i$. Otherwise, $x'_v = x_v$. 
\item $x'_{u, v} = x'_u$ ($v \in S_i$), $x'_v$ ($u \in S_i$), or $x_{u, v}$ otherwise. 
\end{itemize}
It is easy to check that it is a feasible solution to $\SA(S_{i + 1})$ (intuitively, {\em rounding up} only helps feasibility), 
so its value is 
\[
|S_i| + \sum_{v \notin S_i, x_v < \frac{1}{\alpha}} x_v \geq |S_i| + \Frac_{i + 1}, 
\]
which implies 
\[
\Frac_i = \sum_{v \notin S_i, x_v \geq \frac{1}{\alpha}} x_v + 
\sum_{v \notin S_i, x_v < \frac{1}{\alpha}} x_v 
\geq \frac{1}{\alpha} (|S_{i+1}| - |S_i|) + \Frac_{i + 1}.
\]
Finally, we have 
\begin{align*}
  |S_l| - |S_i| \\
= (|S_l| - |S_{i + 1}|) + (|S_{i + 1}| - |S_i|) \\
\leq \alpha \Frac_{i + 1} + (|S_{i + 1}| - |S_i|) \\
\leq \alpha \Frac_{i},
\end{align*}
where the first inequality follows from the induction hypothesis. This completes the induction.
\end{proof}

\begin{lemma}
After the termination, every vertex has degree at most $2k$ in the subgraph induced by $V_G \setminus S$.
\end{lemma}
\begin{proof}
We prove that at least one vertex is added to $S$ if the subgraph induced by $V_G \setminus S$ has a vertex of degree more than $2k$. 
Fix one such iteration, and let $S_1$ and $S_2$ be $S$ before and after the iteration respectively. 
Let $G'$ be the subgraph of $G$ induced by $V_G \setminus S_1$. 
If the subgraph induced by $V_G \setminus S_2$ does not have any vertex with degree more than $2k$, we are done. Otherwise, fix one such vertex $u \in V_G \setminus S_2$. 
Note that the degree of $u$ in $G'$ is also more than $2k$. 

We show that at least one neighbor $v$ of $u$ satisfies $v \notin S_1$ but $v \in S_2$. 
Let $x$ be the optimal fractional solution to $\SA(S_1)$ and consider the following constraint for $u$. 
\[
\sum_{v : (u, v) \in E_G} (x_v - x_{u, v}) \geq (\Deg(u) - k + 2)(1 - x_u).
\]
Let $\Nbr(u)$ and $\Nbr'(u)$ be the set of neighbors of $u$ in $G$ and $G'$ respectively, and $\Deg'(u) = |\Nbr'(u)|$. Note that $\Nbr'(u) = \Nbr(u) \setminus S_1$, and for $v \in \Nbr(u) \cap S_1$, $x_v = 1$ and $x_{u, v} = x_u$. Therefore, the above constraint is equivalent to 
\begin{align*}
& \sum_{v : \Nbr(u) \cap S_1} (1 - x_{u}) + 
\sum_{v : \Nbr'(u)} (x_v - x_{u, v}) \geq (\Deg(u) - k + 2)(1 - x_u) \\
\Leftrightarrow &
\sum_{v : \Nbr'(u)} (x_v - x_{u, v}) \geq (\Deg'(u) - k + 2)(1 - x_u).
\end{align*}
The fact that $u \notin S_2$ implies that $x_u < \frac{1}{\alpha}$, which implies
\begin{align*}
& \sum_{v \in \Nbr'(u)} x_v  \\
\geq & \sum_{v \in \Nbr'(u)} (x_v - x_{u, v}) \geq (1 - \frac{1}{\alpha})(\Deg'(u) - k) 
= & (1 - \frac{1}{\alpha}) \Deg'(u) (1 - \frac{k}{\Deg'(u)}).
\end{align*}
Therefore, there is one $v \in \Nbr'(u)$ with $x_v \geq (1 - \frac{1}{\alpha}) (1 - \frac{k}{\Deg'(u)}) \geq \frac{9}{10}\cdot\frac{1}{2} > \frac{1}{\alpha}$. $v$ satisfies $v \notin S_1$ but $v \in S_2$. 
\end{proof}

\medskip \noindent {\bf Constrained Set Multicover.}
The first phase returns a set $S$ whose size is at most $\alpha$ times the optimal solution and the subgraph induced by $V_G \setminus S$ has maximum degree at most $2k$. As above, let $G'$ be the subgraph induced by $V_G \setminus S$, $\Nbr(u), \Nbr'(u)$ be the neighbors of $u$ in $G$ and $G'$ respectively, and $\Deg(u) = |\Nbr(u)|, \Deg'(u) = |\Nbr'(u)|$.
The remaining task is to find a small subset $F \subseteq V_G \setminus S$ such that the subgraph of $G'$ (and $G$) induced by $V_G \setminus (S \cup F)$ has no vertex of degree at least $k - 1$. 
We reduce the remaining problem to the {\em Constrained Set Multicover} problem defined below. 

\begin{definition}
Given an set system $U = \{ e_1, ..., e_n \}$, a collection of subsets $\calc = \{ C_1, ..., C_m \}$, and a positive integer $r_e$ for each $e \in U$, 
the Constrained Set Multicover problem asks to find the smallest subcollection (each set must be used at most once) such that
each element $e$ is covered by at least $r_e$ times.
\end{definition}

Probably the most natural greedy algorithm does the following:
\begin{itemize}
\item Pick a set $C$ with the largest cardinality (ties broken arbitrarily). 
\item Set $r_e \leftarrow r_e - 1$ for $e \in C$. If $r_e = 0$, remove it from $U$. For each $C \in \calc$, let $C \leftarrow C \cap U$.
\item Repeat while $U$ is nonempty.
\end{itemize}

Constrained Set Cover has the following standard LP relaxation, and Rajagopalan and Vazirani~\cite{RV98} showed that the greedy algorithm gives an integral solution whose value is at most $H_d$ (i.e. the $d$th harmonic number) times the optimal solution to the LP, where $d$ is the maximum set size. 
\begin{align*}
\mbox{minimize  } \qquad & \sum_{C \in \calc} z_C & \\
\mbox{subject to   }\qquad  
& \sum_{C : e \in C} z_C \geq r_e & e \in U \\
& 0 \leq z_C \leq 1 & C \in \calc
\end{align*}

Our remaining problem, $k$-Star Transversal on $G'$, can be thought as an instance of Constrained Set Cover in the following way:
$U := \{ u \in V_G \setminus S : \Deg'(u) \geq k - 1 \}$ with $r_u := \Deg'(u) - k + 2$, 
and for each $v \in V_G \setminus S$, add $\Nbr'(v) \cap U$ to $\calc$. 
Intuitively, this formulation requires at least $r_u$ neighbors be picked in the transversal whether $u$ is picked or not. 
This is not a valid reduction because the optimal solution of the above formulation can be much more than the optimal solution of our problem. 
However, at least one direction is clear (any feasible solution to the above formulation is feasible for our problem), and it suffices to show that the above LP admits a solution whose value is close to the optimum of our problem. 
The LP relaxation of the above special case of Constrained Set Cover is the following:
\begin{align*}
\mbox{minimize  } \qquad & \sum_{v \in V_G \setminus S} z_v & \\
\mbox{subject to   }\qquad  
& \sum_{v : v \in \Nbr'(u)} z_v \geq \Deg'(u) - k + 2 & u \in U \\
& 0 \leq z_v \leq 1 & v \in V_G \setminus S
\end{align*}

Consider the last iteration of the first phase where we solved $\SA(S)$. Let $x$ be the optimal solution to $\SA(S)$ and $\Frac := \sum_v x_v - |S|$. 
Note that $x_v < \frac{1}{\alpha}$ when $v \notin S$. Define $\{ y_v \}_{v \in V \setminus S}$ such that $y_v := 2x_v$. 
\begin{lemma}
$\{ y_v \}$ is a feasible solution to the above LP for Constrained Set Cover.
\end{lemma}
\begin{proof}
By construction $0 \leq y_v < \frac{2}{\alpha}$, so it suffices to check for each $u \in U$, 
\[
\sum_{v : v \in \Nbr'(u) } y_v \geq \Deg'(u) - k + 2.
\]
Fix $u \in U$. Recall that Sherali-Adams constraints on $x$ imply that
\begin{align*}
& \sum_{v : \Nbr'(u)} (x_v - x_{u, v}) \geq (\Deg'(u) - k + 2)(1 - x_u) \\
\Rightarrow & \sum_{v : \Nbr'(u)} x_v \geq (\Deg'(u) - k + 2)(1 - x_u) \\
\Rightarrow & \sum_{v : \Nbr'(u)} 2 x_v \geq \Deg'(u) - k + 2,
\end{align*}
where the last line follows from the fact that $1 - \frac{1}{\alpha} > \frac{1}{2}$. 
\end{proof}
Therefore, Constrained Set Cover LP admits a feasible solution of value $2\Frac$, and the greedy algorithm gives a $k$-Star Transversal $F$ with $|F| \leq 2 \cdot  \Frac \cdot  H_{2k}$. Since $\Frac$ is at most the size of the optimal $k$-Star Transversal for $G'$ (and clearly $G$), $|S \cup F|$ is at most $O(\log k)$ times the size of the smallest $k$-Star Transversal of $G$.

\section{Hardness for Longer Cycles and Connection to FVS}
\label{sec:cycles}

We introduce several notations convinient for cycles. Given an integer $k$ and $i$, let $(i)$ denote the integer in $[k]$ such that $i = (i)\mod k$ --- the choice of $k$ will be clear in the context. 
Recall that we use superscripts $v^i$ and $e^i$ to indicate a vertex and an edge of a graph, respectively. 
In some cases in this section, a vertices is represented as a vector (i.e. in $n$-dimensional hypercube, $V = \{ 0, 1 \}^n$ and each vertex $v = (v_1, ..., v_n)$ is a $n$-dimensional vector). A subscript $v_i$ is used to denote the $i$th coordinate of $v$ in this case. 


\subsection{Proof of Theorem~\ref{thm:random}}
We prove Theorem~\ref{thm:random}, which improves Theorem~\ref{thm:main_v} in the sense that in the completeness case, a small subset $F \subseteq V_G$ intersects not only cycles of length exactly $k$, but also all cycles of length $3, 4, ..., O(\frac{\log n}{\log \log n})$.
The reduction and the soundness analysis are exactly the same. 
We show the following lemma for the completeness case which is again almost identical to Lemma~\ref{lem:random_completeness}, but carefully keeps track of parameters to consider cycles of increasing length. 

\begin{lemma}
Suppose $P$ has a vertex cover $C$ of measure $c$. For any $\epsilon > 0$, with probability at least $3/4$, there exists a subset $F \subseteq V_G$ of measure at most $c + \epsilon$ such that the induced subgraph $V_G \setminus F$ has no cycle of length $O(\frac{\log n}{\log \log n})$. The constant hidden in $O$ depends on $k, \epsilon$ and the degree $d$ of $P$. 
\end{lemma}

\begin{proof}
Let $F = C \times [B]$.
We consider the expected number of cycles that avoid $F$ and argue that a small fraction of additional vertices intersect all of these cycles.
Let $k'$ be the length of a purported cycle. Choose $k'$ vertices $(v^1, l^1), \ldots, (v^{k'}, l^{k'})$ which satisfy
\begin{itemize}
\item $v^1 \in V_P$ can be any vertex.
\item $l^1, \ldots, l^{k'} \in B$ can be arbitrary labels. 
\item For each $1 \leq i < k'$, there must be a hyperedge $e = ( u^1, \ldots, u^k )$ and $j \in [k]$ such that 
($v^i = u^j$ and $v^{i+1} = u^{(j+1)}$) or 
($v^i = u^{(j+1)}$ and $v^{i+1} = u^{j}$). Equivalently, there are edges between $\mathsf{\mathsf{cloud}}(v^i)$ and $\mathsf{\mathsf{cloud}}(v^{i+1})$. 
\end{itemize}
There are $n$ possible choices for $v^1$, $B$ choices for each $l^i$, and $2d$ choices for each $v^i \, (i > 1)$ (there are at most $d$ hyperedges containing one vertex, and for each canonical cycle, there are two possibilities to choose a neighbor). 
The number of possibilities to choose such $(v^1, l^1), \ldots, (v^{k'}, l^{k'})$ is bounded by $n (2d)^{k'-1} B^{k'}$. 
Note that no other $k'$-tuple of vertices can form a cycle. Further discard the tuple when two vertices are the same (the resulting cycle is not simple and its simple pieces will be considered for smaller $k'$). 

We calculate the probability that $((v^1, l^1), \ldots, (v^{k'}, l^{k'}))$ forms a cycle (i.e. all $k'$ edges exist) that does not intersect $F$. For a set of purported edges, we say that this set can be {\em covered by a single canonical cycle} if one copy of canonical cycle can contain all $k'$ edges with nonzero probability. 
Suppose that all $k'$ edges in the purported cycle can be covered by a single canonical cycle. 
It is only possible when $k' = k$ and there is a hyperedge $e$ such that after an appropriate shifting, $e = ( v^1, \ldots, v^{k} )$ (recall that $e$ is considered to be an ordered $k$-tuple). 
In this case, $((v^1, l^1), \ldots, (v^{k}, l^{k}))$ intersects $F$ (right case of Figure~\ref{fig:partition}). 
When $k'$ edges of the purported cycle have to be covered by more than one canonical cycle, some vertices must be covered by more than one canonical cycle, and each canonical cycle covering the same vertex should give the same label to that vertex. This redundancy makes it unlikely to have all $k'$ edges exist at the same time (left case of Figure~\ref{fig:partition}). 
The below claim, similar to Claim~\ref{claim:numberofcycles} but desginated for cycles to obtain better parameters, formalizes this intuition.

\begin{claim}
Suppose that $((v^1, l^1), \ldots, (v^{k'}, l^{k'}))$ cannot be covered by a single canonical cycle. Then the probability that it forms a cycle is at most $k'(\frac{adk'}{B})^{k'}$.
\end{claim}
\begin{proof}
Fix $2 \leq p \leq k'$. Partition $k'$ purported edges into $p$ nonempty {\em groups} $I_1, \ldots, I_p$
such that each group can be covered by a single canonical cycle. 
There are at most $p^{k'}$ possibilities to partition.
For each $v \in V_P$, there are at most $d$ hyperedges containing $v$ and at most $aBd$ canonical cycles intersecting $\mathsf{\mathsf{cloud}}(v)$. 
Therefore, all edges in one group can be covered simultaneously by at most $aBd$ copies of canonical cycles. 
There are at most $(aBd)^p$ possibilities to assign a canonical cycle to each group. 
Assume that one canonical cycle is responsible for exactly one group. 
This is without loss of generality since if one canonical cycle is responsible for many groups, we can merge them and this case can be dealt with smaller $p$.

Focus on one group $I$ of purported edges, and one canonical cycle $L$ which is supposed to cover them. Let $I' \subseteq V_G$ be the set of vertices which are incident on the edges in $I$. 
Suppose $L = ((u^1, l'^1), \ldots, (u^k, l'^k))$, which is produced by a hyperedge $f = (u^1, \ldots, u^k) \in E_P$. 
We calculate the probability that $L$ contains all edges in $I$ over the choice of labels $l'^1, \ldots, l'^k$ for $L$. 
One necessary condition is that 
\[ \left\{ v | (v, l) \in I' \mbox{ for some } l \in [B] \right\} \] 
(i.e., the set $I'$ projected to $V_P$) is contained in $f$. Otherwise, some vertices of $I'$ cannot be covered by $L$. 
Another necessary condition is $v^i \neq v^j$ for any $(v^i, l^i) \neq (v^j, l^j) \in I'$.
Otherwise ($(v, l^i), (v, l^j) \in I'$ for $l^i \neq l^j$), since $L$ gives only one label to each vertex in $f \subseteq V_P$, $(v, l^i)$ and $(v, l^j)$ cannot be contained in $L$ simultaneously. Therefore, we have a nice characterization of $I'$: It consists of at most one vertex from the cloud of each vertex in $f$.

Now we make a crucial observation that $|I'| \geq |I| + 1$.
This is because $I$ is a proper subset of the edges that form a simple cycle. 
Formally, in the graph with vertices $I'$ and edges $I$, the maximum degree is at most 2, and there are at least two vertices of degree 1. 
The probability that $L$ contains $I$ is at most the probability that for each $(v^i, l^i) \in I'$, $l^i$ is equal to the label $L$ assigns to $v^i$, which is $B^{-|I'|} \leq B^{-|I| - 1}$. 

We conclude that for each partition, the probability of having all the edges is at most 
\[ (aBd)^p \prod_{q= 1}^p{B^{-|I_q|-1}}=  \frac{(aBd)^p}{B^{k' + p}} = \frac{(ad)^p}{B^{k'}} \ .\]
The probability that $((v^1, l^1), \ldots, (v^{k'}, l^{k'}))$ forms a cycle is therefore bounded by 
\[  \sum_{p=2}^{k'} \ p^{k'} \frac{(ad)^p}{B^{k'}} \leq k' \bigl( \frac{adk'}{B}\bigr)^{k'} \ . \]
\end{proof}

Therefore, the expected number of cycles of length $k'$ that avoid $F$ is bounded by
$n(2d)^{k'-1}B^{k'} \cdot k' (\frac{adk'}{B})^{k'} \leq n (Rk')^{k'}$
where $R$ is a constant depending only on $a$ and $d$ (both are independent of $k'$). With probability at least $3/4$, the number of such cycles of length up to $k'$ is at most $4 n (Rk')^{k'+1}$.
Let $B \geq \frac{4 (Rk')^{k'+1}}{\epsilon}$. Then these cycles can be covered by at most $\epsilon n B = \epsilon N$ vertices. 
If $k' = \frac{\log n}{\log \log n}$, then $k'^{k'} = \exp(k' \log k')$ is also $o(n)$, we can take $B$ linear in $n$ and $k' \geq \Omega(\frac{ \log N}{\log \log N})$. 
\end{proof}

\subsection{Hardness of $k$-HVC with Bounded Degree and Density}
\label{subsec:hvc}
In this subsection, we observe the implicit properties of the best known hardness result for $k$-HVC~\cite{DGKR05} such as {\em bounded degree} and {\em density}. Only bounded degree is needed for our main Theorem~\ref{thm:main_v} as well as its extension Theorem~\ref{thm:random} for cycles proved in the previous subsection, while another derandomized proof of hardness in the next subsection requires density property as well.
The following is the theorem explicitly stated in~\cite{DGKR05}. 

\begin{theorem} [\cite{DGKR05}]
[Restatement of Theorem~\ref{thm:dgkr_body}]
Given a $k$-uniform hypergraph ($k \geq 3$) and $\epsilon > 0$, it is NP-hard to tell apart the following cases:
\begin{itemize}
\item Completeness: There exists a vertex cover of measure $c := \frac{1 + \epsilon}{k - 1}$.
\item Soundness: Every vertex cover has measure at least $s := 1 - \epsilon$.
\end{itemize}
Therefore, it is NP-hard to approximate $k$-HVC within a factor $k - 1 + 2\epsilon$. 
\end{theorem}

In some cases, we need a fact that a given hypergraph $P$ has small degrees (only function of $k$ and $\epsilon$) as well as the following additional {\em density} property in the soundness case. In a hypergraph, we define the degree of a vertex to be the number of hyperedges containing it. 
\begin{itemize}
\item The maximum degree of $P$ is bounded by $d$. 
\item In the soundness case above, every set of measure at least $\delta > 0$ contains $\rho > 0$ fraction of hyperedges in the induced subgraph. 
\end{itemize}
For example, the NP-hardness of $k$-HVC with $c = \frac{3}{k}$, $s = 1 -\frac{1}{k}$, $d = 2^{k^{\beta}}$, $\delta = \frac{2}{k}$, and $\rho = \frac{1}{k2^{k^{\beta}}}$ for some $\beta$ is made explicit in~\cite{CGHKKKN05}.
However, careful examination of other results, especially that of~\cite{DGKR05}, yields a better result.
\begin{theorem} [\cite{DGKR05}]
\label{thm:dgkr}
For any rational $\epsilon > 0$, Theorem~\ref{thm:dgkr_body} holds with
$d = O(1), \delta > 0, \rho > 0$ are some constant depending on $k$ and $\epsilon$.
\end{theorem}
\begin{proof}
Theorem 4.1 of~\cite{DGKR05} requires a multi-layered PCP with parameters $l$ (number of layers) and $R$ (number of labels), which both depend on $k$ and $\epsilon$. Note that in the original Raz verifier, the degree $d_R$ is a function of $R$. 
Given a Raz verifier which consists of a bipartite graph $G = (V_G, E_G)$ such that $V_G = Y \cup Z$, Theorem 3.3 yields a multilayered PCP where variables of layer $i$ are of the form $(z_1, \ldots, z_{i}, y_{i+1}, \ldots, y_l)$ where $z_j \in Z$ and $y_j \in Y$. The number of labels for any vertex is bounded by $R^l$. For $i < j$, there exists a constraint between $(z_1, \ldots, z_{i}, y_{i+1}, \ldots, y_l)$ and $(z'_1, \ldots, z'_{j}, y'_{j+1}, \ldots, y'_l)$ if and only if
\begin{itemize}
\item $z_q = z'_q$ where $q \leq i$. 
\item $y_q = y'_q$ where $q > j$. 
\item $(y_q, z'_q) \in E_G$ for $i < q \leq j$. 
\end{itemize}
Therefore, the degree is at most $l(d_R)^l$, which is still a function of $k$ and $\epsilon$. After the reduction from a multilayered PCP to a weighted hypergraph, the degree of each vertex is still bounded by a function of $k$ and $\epsilon$, since each variable of the PCP is replaced by at most $2^{R^l}$ vertices and each PCP constraint is replaced by at most $2^{kR^l}$ hyperedges.

Given such a weighted instance, we convert it to an unweighted instance by duplicating vertices according to their weights. 
The weight of each vertex in the $i$th layer is of the form
\[
\frac{1}{l |X_i|} p^{r}(1-p)^{R_i - r}
\]
where $X_i = |Z_i|^i |Y_i|^{l-i}$ is the set of vertices in the $i$th layer, $R_i = R^{O(l)}$ is the number of labels in $i$th layer, $p = 1 - \frac{1}{k - 1 - \epsilon}$, and $0 \leq r \leq R_i$.
The original paper set the weight as above so that the sum of weights becomes 1. Multiply weight of each vertex by $|Y_i|^l$ so that the weight of each vertex in the $i$th layer is of the form 
\[
\frac{1}{l} \left(\frac{|Y_i|}{|Z_i|}\right)^i p^{r}(1-p)^{R_i - r}
\]

 Let $\alpha$ be a rational that divides both $p$ and $1-p$ with both quotients bounded. Then $\alpha^{R^l}$ divides any $p^{r}(1-p)^{R_i - r}$ as well with quotient bounded by a function of $\epsilon$ and $k$. 
Therefore, if we set the minimum weight to be
\[
\frac{1}{l}\cdot \frac{|Y_i|}{|Z_i|} \cdot \alpha
\]
the weight of each vertex must be divisible by the minimum weight, and the quotient will be bounded by a function of $k$ and $\epsilon$. We replace each weighted vertex by (weight / minimum weight) number of unweighted vertices, and for each hyperedge $(v^1, \ldots, v^k)$, add all hyperedges $(u^1, \ldots, u^k)$ where $u^i$ is a copy $v^i$. Since each quotient and the original degree of the weighted instance are bounded by a function of $k$ and $\epsilon$, so is the degree of the unweighted instance. 

Now we have an unweighted problem with completeness $c$, soundness $s$, and degree bounded by $d$. 
Let $\delta = 2(1 - s)$. 
Suppose in soundness case, we have $1 - \delta$ fraction of vertices cover more than $1 - \frac{k(1 - s)}{d}$ fraction of hyperedges.
Cover the remaining hyperedges with one vertex each. Since $|E_P| \leq \frac{d}{k}|V_P|$, this process requires less than $\frac{k(1 - s)}{d} \cdot \frac{d}{k} = 1 - s$ fraction of vertices, and we have a vertex cover of measure less than $1 - \delta + (1 - s) = s$. This contradicts the original soundness, so any $\delta := 2(1-s)$ fraction of vertices should contain at least $\rho := \frac{k\delta}{2d}$ fraction of edges, both depending only on $k$ and $\epsilon$. 
\end{proof}

\subsection{Labeling Gadget}
\label{subsec:labeling_gadget}

We now give another proof of hardness of $k$-Cycle Transversal. 
It is weaker than Theorem~\ref{thm:random} in the sense that a small subset intersects cycles of length at most $O(\log \log n)$ in the completeness case while in Theorem~\ref{thm:random}, we are able to intersect cycles of length $O(\frac{\log n}{\log \log n})$). 
However, it has an advantage of being derandomized so that the result assumes only $\mathsf{P} \neq \mathsf{NP}$ instead of $\mathsf{NP} \not \subseteq \mathsf{BPP}$.
It crucially uses the fact that graphs are directed, so we hope that further improvements on this technique will allow more progress on hardness of FVS, which is hard to approximate only on directed graphs. 

\begin{theorem}
Fix an integer $k \geq 3$ and $\epsilon \in (0,1)$. 
Given a directed graph $G = (V_G, E_G)$, unless $\mathsf{NP} \subseteq \mathsf{P}$, there is no polynomial time algorithm that distinguishes between the following two cases.
\label{thm:labeling}
\begin{itemize}
\item Completeness: There exists $F \subseteq V_G$ with $\frac{1}{k-1} + \epsilon$ fraction of vertices that intersects every cycle of length at most $O(\log \log n)$ (hidden constant in $O$ depends on $k$ and $\epsilon$).
\item Soundness: Any subset with more than $\epsilon$ fraction of vertices has a cycle of length exactly $k$ in the induced subgraph. 
\end{itemize}
\end{theorem}

\medskip\noindent {\bf Intuition.}
We call this technique {\em labeling gadget}, which explicitly controls the structure of every cycle.
The idea of labeling gadgets to prove hardness of approximation has been used previously to show inapproximability of edge-disjoint paths problem with congestion and directed cut problems~\cite{AZ08,CGKT07,CK09}.

In this work, the labeling gadget is a directed graph $L = (V_L, E_L)$ with roughly the following properties: (i) its girth is $k$, and (ii) every subset of vertices of measure at least $\delta$ has at least one cycle of length $k$.

To highlight the main idea, we introduce a valid reduction from FVS that increases the size of instances exponentially --- the actual proof increases the size polynomially but only works for cycles of bounded length. 
Given a hypergraph $P$ and the labeling gadget $L$, $G = (V_G, E_G)$ is constructed in the following way. 
$V_G = V_P \times V_L^m$, where $m$ is the number of hyperedges in $P$ (say the hyperedges are $e^1,e^2,\dots,e^m$), so that the cloud for each vertex $v \in V_P$ becomes $V_L^m$. 
Each copy of $L$ corresponds to one of the $m$ hyperedges. 
Consider the naive approach introduced earlier where we added $k$ edges for each hyperedge (multiple edges possible), without duplicating vertices. Call this graph $P' = (V_P, E_{P'})$. In $G$, we add an edge from $(v, x^1, \ldots , x^m)$ to 
$(u, y^1, \ldots , y^m)$ if and only if
\begin{itemize}
\item There is an edge $(v, u) \in E_{P'}$ created by a hyperedge $e^i$ for some $i$.
\item $x^j = y^j$ for all $j \neq i$, and 
\item $(x^i, y^i) \in E_L$. 
\end{itemize}
Intuitively, if we want to move from $(u, \ldots)$ to $(v, \ldots)$ where the edge $(u, v) \in E_{P'}$ is created by a hyperedge $e^i$,
then we need to move the $i$th coordinate by an edge of $L$ (other coordinates stay put). 
Once we changed the $i$th coordinate, since $L$ has girth $k$, we have to use an edge formed by $e^i$ at least $k$ times to move $i$th coordinate back to the original solution. 

Suppose $\mathcal{C}=((v^1, \ldots), \cdots , (v^{k'}, \ldots))$ is a cycle in $G$. By the above argument, $(v^1, \ldots, v^{k'})$ is a cycle of $P'$, and must use at least $k$ edges formed by a single hyperedge, say $e^l$.
This is not quite enough to argue that this cycle intersects a vertex cover of $P$ as the same edge of $P'$ that is created by hyperedge $e^i$ may be used multiple times. 
To fix this problem, we color each edge of $L$ by one of $k$ colors and associate a different color to the $k$ edges formed by a hyperedge.
If we ensure the stronger property in the labeling gadget that every cycle of $L$ must be colorful (which implies that the girth is at least $k$), then the cycle $\mathcal{C}= ((v^1, \ldots), \cdots , (v^{k'}, \ldots))$ uses all $k$ edges formed by a single hyperedge, so it must intersect any vertex cover of $P$. 
See Figure~\ref{fig:label} for an example.




\begin{figure}
   \centering
    \includegraphics[height=6cm]{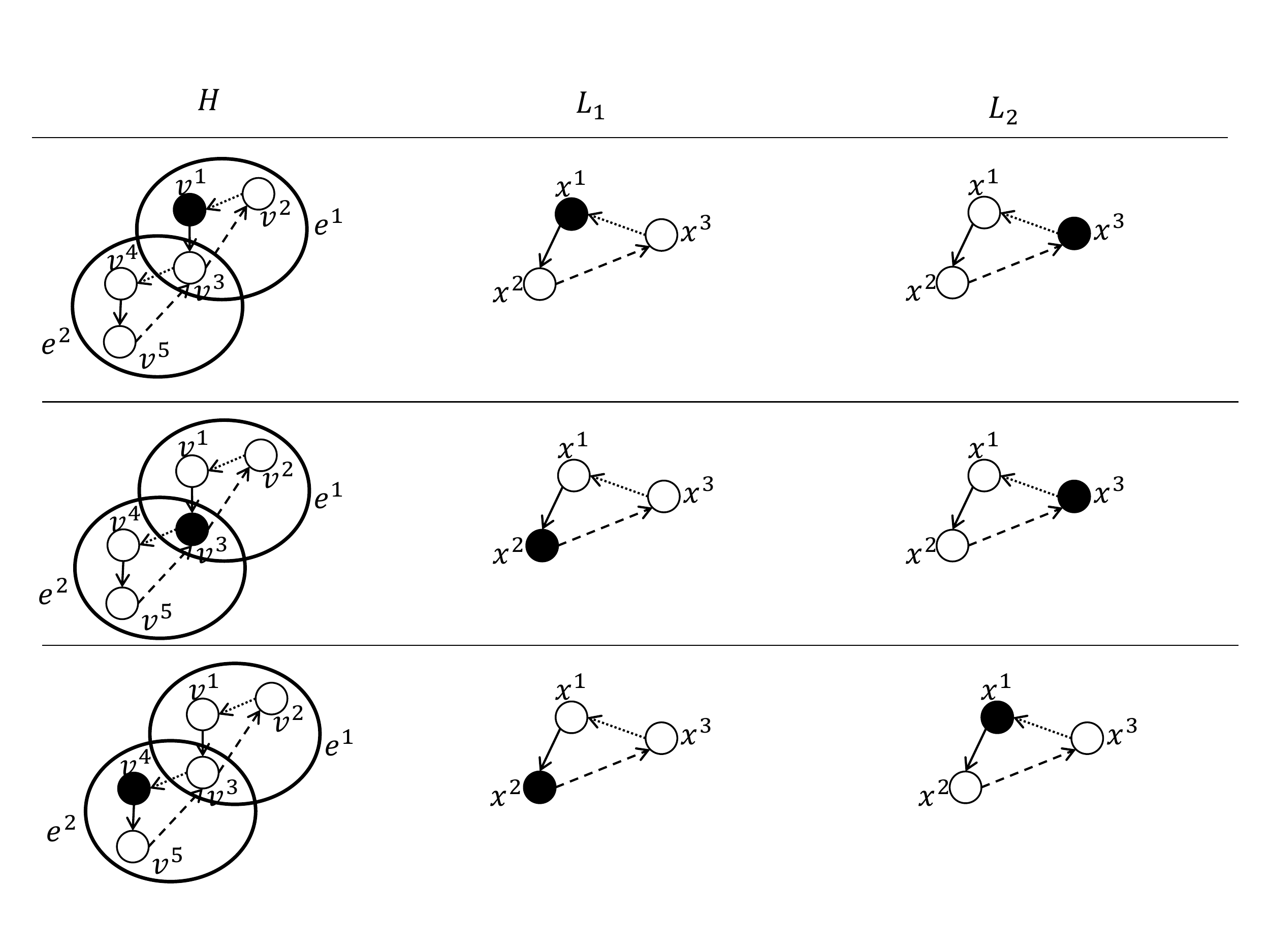}
  \caption{{\small Example with $k = 3$. Each row corresponds to a vertex of $G$ ($(v^1, x^1, x^3)$ in the first row), and each edge of $P$ and $L$ has one of 3 types. From $(v^1, x^1, x^3)$, we used $e_1$ and the solid edge to get to
$(v^3, x^2, x^3)$. The position in $L_2$ stays the same. From $(v^3, x^2, x^3)$, we used $e_2$ and the dotted edge to get to $(v^4, x^2, x^1)$. 
 }}
  \label{fig:label}
\end{figure}

For soundness, given a subset $F \subseteq V_G$ of measure $\delta$, we find a hyperedge $e = ( v^1, \ldots, v^k )$ such that 
$(\bigcap\, \mathsf{\mathsf{cloud}}(v^i)) \cap F$ is large. This follows from averaging arguments and needs a {\em density} guarantee in the soundness case of $k$-HVC. Then we focus on the copy of $L$  associated with $e$, find a colorful $k$-cycle in $L$, and produce the final cycle by combining two cycles $(v^1, v^2), \cdots , (v^k, v^1)$ (from $V_P$) and the colorful cycle in $L$. 

This is a complete and sound reduction from $k$-HVC to the {\em original} FVS problem, except that it blows up the size of the instance exponentially. To get a polynomial time reduction, we compress the construction by coalescing different copies of $L$, retaining only a constant number (dependent on the degree of the original hypergraph) out of the $m$ coordinates.
However, as a result we are not able to control the behavior of long cycles, and we may not intersect all cycles in the completeness case of Theorem~\ref{thm:labeling}. Since we have good
control over the structure of cycles using labeling gadgets, and the only issue is to reduce the size of labels, we hope that  more sophisticated variants of this technique might be able to prove inapproximability of FVS itself.

\noindent{\bf Labeling Gadget.} 
A $(k, \delta)$-labeling gadget is a directed graph $L = (V_L, E_L)$ with each edge colored with a color from $[k]$ that satisfies the following three properties. 
\begin{enumerate}
\item Its girth is exactly $k$.
\item Every cycle has at least one edge for each color.
\item Every subset of vertices of measure at least $\delta$ has at least one cycle $(x^1, x^2, \dots , x^k)$ such that
\begin{itemize}
\item Its length is exactly $k$.
\item After an appropriate shifting, the color of $(x^i, x^{(i+1)})$ is $i$.
\end{itemize}
\end{enumerate}
Let $V_L = [B]^k$, where $B$ will be determined later depending on $\delta$ and $k$. 
For each $1 \leq i \leq k$, and for each $x_1, \dots , x_k$ and $y_i > x_i, y_{(i+1)} > x_{(i+1)}$, we add an edge of color $i$ from 
\[ (x_1, \dots , x_i, y_{(i+1)}, \dots , x_k) \quad \text{to} \quad (x_1, \dots , y_i, x_{(i+1)}, \dots , x_k) \ . \]
 Intuitively, edges of color $i$ {\em strictly increase ith coordinate, strictly decrease $(i+1)$th coordinate, and do not change the others}.

With this construction, properties 1. and 2. can be shown easily. If a cycle uses an edge of color $i$, the $i$th coordinate 
was decreased by using this edge, and the cycle should use at least one edge of color $(i+1)$ to return. The same argument can be applied to color $(i+1)$, $(i+2)$, $\dots$, until the cycle uses all the colors. 
The following lemma shows property 3.

\begin{lemma}
\label{lem:gadget}
For $k \in \mathbb{N}$ and $\delta > 0$, there exists an integer $B := B(k, \delta)$ such that a subset $S \subseteq [B]^k$ with measure at least $\delta$ contains a $k$-cycle that has one edge of each color.
\end{lemma}
\begin{proof}
Fix a subset $S \subseteq [B]^k$ of measure at least $\delta$. 
For each $x \in [B]^k$ and $i \in [k]$, define $\Line(x, i) := \left\{ y \in [B]^k : (y)_j = (x)_j \mbox{ for all } j \neq i \right\}$ to be the axis-parallel line containing $x$ and parallel to the $i$th unit vector $e_i$. Let 
$\Surface_S$ map each directed line to the first point in $S$ that the line hits. Precisely, 
\begin{equation*}
\Surface_S(x, i) := 
\begin{cases}
\argmax_{y \in S \cap \Line(x, i)} (y_i) & S \cap \Line(x, i) \neq \emptyset \\
\emptyset & S \cap \Line(x, i) = \emptyset,
\end{cases}
\end{equation*}
and $S' := \cup_{x, i} \Surface_S(x, i)$. 
There are $k \cdot B^{k - 1}$ lines total ($B$ points for each line), so $|S'| \leq k \cdot B^{k - 1}$.
If $B > \frac{k}{\delta}$, there is an element in $S \setminus S'$. Call this point $(x_1, \dots, x_k)$. 
For any $i \in [k]$, $(x_1, \dots, x_{i - 1}, y_i, x_{i + 1}, x_k)$ is also in $S$ for some $y_i > x_i$.
$((x_1, x_2, \dots , x_{k-1}, y_k)$, 
$(x_1, x_2, \dots , y_{k-1}, x_k)$, $\dots$,
$(x_1, y_2, \dots , x_{k-1}, x_k)$,
$(y_1, x_2, \dots , x_{k-1}, x_k)$,
$(x_1, x_2, \dots , x_{k-1}, y_k))$
is a cycle we wanted.
\end{proof}

\noindent{\bf Reduction. }
We show a reduction from $k$-HVC to Directed $k$-Cycle Transversal, proving Theorem~\ref{thm:labeling}.
Fix $k$ and let $c := \frac{1 + \epsilon}{k - 1}, s := 1 - \epsilon, d, \delta, \rho$ be the parameters we have from Theorem~\ref{thm:dgkr}. 

Let $k'$ be the maximum length of cycles that we want to intersect in the completeness case, which will be determined later.
Let $L = (V_L, E_L)$ be a $(k, \rho \delta)$-labeling gadget. 
We are given a hypergraph $P = (V_P, E_P)$ with the maximum degree $d$. Since each vertex has a degree at most $d$, each hyperedge shares a vertex with at most $dk$ other hyperedges. Consider a graph $P' = (V_{P'}, E_{P'})$ where $V_{P'} = E_P$ and there exists an edge between $e$ and $f$ if and only if they intersect. Define the distance between two hyperedges $e$ and $f$ to be the minimum distance between $e$ and $f$ in $P'$.  
The maximum degree of $P'$ is bounded by $dk$, and for each $e \in V_{P'}$, there are at most $(dk)^{k'}$ neighbors within distance $k'$.
Therefore, each hyperedge can be colored with $d' = (dk)^{k'} + 1$ colors so that two hyperedges within distance $k'$ are assigned different colors. To distinguish it from the coloring of $L$, we call the former {\em outer coloring} and the latter {\em inner coloring}. We use letters $u, v$ to denote the vertices of $V_P$, $x, y$ for $V_L$, and $a, b$ for $(V_L)^{d'}$. Furthermore, since some vertices are indexed by a vector, we use superscripts to denote different vertices (e.g. $x^1, x^2 \in V_L$) and subscripts to denote different coordinates of a single vertex (e.g. $x = (x_1, \ldots, x_k)$). 

Our reduction will produce a directed graph $G = (V_G, E_G)$ where $V_G = V_P \times (V_L)^{d'} = V_P \times ([B]^k)^{d'}$. 
The number of vertices (from $P$ to $G$) is increased by a factor of $|V_L|^{d'} = |V_L|^{(dk+1)^{k'}}$. 
Since $|V_L|$ and $dk+1$ only depend on $k$, this quantity is polynomial in $|V_P|$ if $k' = O(\log \log |V_P|)$. 
The edges of $G$ are constructed as the following: 
\begin{itemize}
\item For any $e = ( v^1, \ldots, v^k ) \in E_P$, let $q \in [d']$ be its (outer) color. 
\item For any $i \in [k]$,
\item For any $x, y \in V_L$ such that $(x, y) \in E_L$ with inner color $i$, 
\item For any $a \in (V_L)^{d'}$,
\item We put an edge $(v^i, a_{q \mapsto x})$ to $(v^{(i+1)}, a_{q \mapsto y})$ with outer color $q$ and inner color $i$, where $a_{q \mapsto x}$ means that the $q$th {\em outer} coordinate of $a$ (which is an element of $V_L$) is replaced by $x$. 
\end{itemize}
For a vertex $(v, a) \in V_G$, consider $a = (x^1, \ldots, x^{d'})$ as a label which is a $d'$-dimensional vector and each coordinate $x^i$ corresponds to a vertex of $L$. Following one edge with outer color $q$ changes only $x^q$ (according to $L$), while leaving the other coordinates unchanged. Based on this fact, it is easy to prove the following lemmas.
\begin{lemma}
$G$ has girth at least $k$.
\end{lemma}
\begin{proof}
From the above discussion, each edge of $G$ acts like an edge for exactly one copy of $L$ and acts like a self-loop for the other copies of $L$. If $((v^1, a^1), \ldots, (v^l, a^l))$ is a cycle in $G$, then each coordinate of $a^i$ is a cycle in $L$ as well. Since $L$ has girth $k$, $G$ also has girth at least $k$. 
\end{proof}
\begin{definition}[Canonical cycles]
For any hyperedge $e = ( v^1, \ldots, v^k )$ of $P$ with outer color $q$, for any cycle $x^1 , \ldots, x^k$ of $L$ such that $(x^i, x^{(i+1)})$ is colored $i$ for $i = 1, 2, \ldots, k$, and for any $a \in (V_L)^{d'}$, $((v^1, a_{q \rightarrow x^1}), \ldots , (v^k, a_{q \rightarrow x^k}))$ is also a cycle of $G$ of length exactly $k$. Call such cycles {\em canonical}. 
\end{definition}

\begin{lemma}
Suppose $k \leq l \leq k'$, and $((u^1, a^1),  \ldots, (u^{l}, a^{l}))$ be a cycle. 
Then, there exists a hyperedge $e$ such that $e \subseteq \left\{ u^1, \ldots, u^l \right\}$.
\label{lem:k_canonical}
\end{lemma}
\begin{proof}
Let one of the edges of the cycle have outer color $q$. 
By the properties of $L$ (corresponding to outer color $q$), 
for each $i \in [k]$, there must be an edge with outer color $q$ and inner color $i$.
Since the distance between two hyperedges with the same outer color is at least $k'$,
every edge with outer color $q$ must be from the same hyperedge, say $e = ( v^1, \ldots, v^k )$. 

By the property 2. of the labeling gadget corresponding to outer color $q$ (equivalently hyperedge $e$), for every inner color $j$, $((u^1, a^1),  \ldots, (u^{l}, a^{l}))$ must use an edge with inner color $j$ and outer color $q$.
Notice that if $((u^i, a^i), (u^{(i+1)}, a^{(i+1)}))$ is with outer color $q$ and inner color $j$, $u^i = v^j$ and $u^{(i+1)} = v^{(j+1)}$. Therefore, $e \subseteq \left\{ u^1, \ldots, u^l \right\}$. 
\end{proof}
%
\noindent{\bf Completeness.}
\begin{lemma}
Recall that $k' = O(\log \log |V_G|)$. 
If $P$ has a vertex cover of measure $c$, $G$ has a $k'$-cycle transversal of measure $c$.
\label{lem:labeling_completeness}
\end{lemma}
\begin{proof}
Let $C \subseteq V_P$ be such that it has measure $c$ and intersects every hyperedge $e \in E_P$.
Let $F = C \times (V_L)^{d'} \subseteq V_G$. It is clear that $F$ has measure $c$. 
We argue that $F$ indeed intersects every cycle of length at most $k'$. 
For every cycle $((u^1, a^1), \ldots, (u^l, a^l))$ of length $k \leq l \leq k'$, by Lemma~\ref{lem:k_canonical}, there exists a hyperedge $e = ( v^1, \ldots, v^k )$ such that $e \subseteq \left\{ u^1, \ldots, u^l \right\}$. 
Since $C$ is a vertex cover for $P$, there exists $v^i \in C$, so $F \supseteq v^i \times (V_L)^{d'}$ intersects this cycle. 
\end{proof}
\noindent{\bf Soundness.}
\begin{lemma}
\label{lem:soundness-labeling}
If every subset of $V_P$ with measure at least $\delta$ contains a $\rho$ fraction of hyperedges in the induced subgraph, 
every subset of $V_G$ with measure $2\delta$ contains a canonical cycle. 
\label{lem:labeling_soundness}
\end{lemma}
\begin{proof}
Let $I \subseteq V_G$ has measure at least $2 \delta$. 
For $a \in (V_L)^{d'}$, we let $\mathsf{slice}(a) := V_P \times a$ to be the copy of $V_P$ associated with $a$. 
Let $A = \left\{ a \in (V_L)^{d'} : \mu_{P} (\mathsf{slice}(a) \cap I) \geq \delta  \right\}$. An averaging argument shows that $\mu_{(V_L)^{d'}} (A) \geq \delta$.
By the soundness property (with density) of $k$-HVC, for each $a \in A$, $\mathsf{slice}(a) \cap I \subseteq V_P$ contains at least $\rho$ fraction of hyperedges. Therefore, if we consider the product space $E_P \times V_L^{d'}$, at least $\rho \delta$ fraction of tuples $(e, a)$ satisfy $e \subseteq \mathsf{slice}(a) \cap I$. 

By an averaging argument with respect to $E_P$, we can conclude that there exists a hyperedge $e = (v^1, \ldots, v^k)$ such that $\rho \delta$ fraction of $a = (x^1, \ldots, x^{d'}) \in (V_L)^{d'}$ satisfies $e \subseteq \mathsf{slice}(a) \cap I$. Without loss of generality, assume that its outer color is 1. 
Another averaging argument with respect to $x^2, \ldots, x^{d'}$ shows that there exists $(y^2, \ldots, y^{d'})$ such that 
$X := \left\{x \in L ~\mid ~ e \subseteq \mathsf{slice}((x, y^2, \ldots, y^{d'})) \cap I \right\}$ satisfies $\mu_L(X) \geq \rho \delta$.

Since $L$ is a $(k, \rho \delta)$-labeling gadget, there exists a cycle $(x^1, \ldots, x^k) \subseteq X$ such that $(x^i, x^{i+1})$ is colored with $i$. 
Our final cycle of $G$ consists of 
\[
((v^i, x^i, y^2, \ldots, y^{d'}), (v^{(i+1)}, x^{(i+1)}, y^2, \ldots, y^{d'}))
\]
for each $i \in [k]$. Note that $(v^i, x^i, y^2, \ldots, y^{d'}) \in I$ for each $i$ since by the definition of $X$, for each $x^i \in X$, $e \subseteq \mathsf{slice}(x^i, y^2, \ldots, y^{d'}) \cap I$. The edge 
\[
((v^i, x^i, y^2, \ldots, y^{d'}), (v^{(i+1)}, x^{(i+1)}, y^2, \ldots, y^{d'}))
\]
exists by the construction. 
\end{proof}
\end{document}